\documentclass[aps,twocolumn,showpacs,pra,superscriptaddress,floatfix,longbibliography]{revtex4-1}

\usepackage{amsmath,amssymb,amsmath,amsthm,bm,graphicx,xcolor}
\usepackage[breaklinks]{hyperref}
\hypersetup{colorlinks=true}

\usepackage[T1]{fontenc}       
\usepackage{subcaption}

\newcommand{\word}[1]{\quad\mbox{#1}\quad} 

\newcommand{\N}{{\mathbb{N}}}
\newcommand{\R}{{\mathbb{R}}}

\newcommand{\up}{\uparrow}
\newcommand{\dn}{\downarrow}

\newcommand{\bra}[1]{\mbox{$\langle #1 |$}}

\newcommand{\ket}[1]{\mbox{$| #1 \rangle$}}

\newcommand{\ii}{\bm{i}}   
    
\newcommand{\Tr}{\rm Tr}

\newcommand{\half}{\tfrac{1}{2}}      %% small fraction  1/2
\renewcommand{\H}{\mathcal{H}}
\renewcommand{\word}[1]{\quad\mbox{#1}\quad} %% well-spaced words
\newtheorem{thm}{Theorem}
\newtheorem{lem}[thm]{Lemma}

\begin{document}

\title{Towards a formal definition of static and dynamic electronic correlations}
  
\author{Carlos L. Benavides-Riveros}
\email{carlos.benavides-riveros@physik.uni-halle.de}
\affiliation{Institut f\"ur Physik, Martin-Luther-Universit\"at
Halle-Wittenberg, 06120 Halle (Saale), Germany}

\author{Nektarios N. Lathiotakis}
\affiliation{Theoretical and  Physical Chemistry Institute,
National Hellenic Research Foundation, GR-11635 Athens, Greece}

\author{Miguel A. L. Marques}
\affiliation{Institut f\"ur Physik, Martin-Luther-Universit\"at
Halle-Wittenberg, 06120 Halle (Saale), Germany}


%\begin{tocentry}
%
%Some journals require a graphical entry for the Table of Contents.
%This should be laid out ``print ready'' so that the sizing of the
%Inside the \texttt{tocentry} environment, the font used is Helvetica
%8\,pt, as required by \emph{Journal of the American Chemical
%The surrounding frame is 9\,cm by 3.5\,cm, which is the maximum
%permitted for  \emph{Journal of the American Chemical Society}
%graphical table of content entries. The box will not resize if the
%content is too big: instead it will overflow the edge of the box.
%
%This box and the associated title will always be printed on a
%separate page at the end of the document.
%
%\end{tocentry}

\begin{abstract}
Some of the most spectacular failures of density-functional 
and Hartree-Fock theories are related to an incorrect 
description of the so-called static electron correlation. 
Motivated by recent progress on the N-representability 
problem of the one-body density matrix for pure states, 
we propose a way to quantify the static contribution to the 
electronic correlation.
By studying several molecular systems we show that our
proposal correlates well with our intuition of static and dynamic 
electron correlation. Our results bring out the paramount importance
of the occupancy of the highest occupied natural spin-orbital
in such quantification. 
\end{abstract}

\pacs{31.15.V-, 31.15.xr, 31.70.-f}

\maketitle


\section{Introduction}
\label{sec:intro}

The concept of correlation and more precisely the idea 
of correlation energy are central in quantum chemistry. 
Indeed, the electron-correlation problem (or how the 
dynamics of each electron is affected by the others) is 
perhaps the single largest source of error in quantum-chemical 
computations \cite{TewKlopperHelgaker}. The success of 
Hartree-Fock theory in providing a workable  
upper bound for the ground-state energy is largely 
due to the fact that a single Slater determinant is 
usually the simplest wave function having the correct 
symmetry properties for a system of fermions. 
Since the description of interacting fermionic systems 
requires multi-determinantal reference wave functions,
the correlation energy is commonly defined as the 
difference between the exact ground-state and the 
Hartree-Fock energy  \cite{WignerCorr,Lowdin}. Beyond 
Hartree-Fock theory, numerous other methods (such as
configuration interaction or coupled-cluster theory) aim
at reconstructing the part of the energy missing from a
description based on a single-determinantal wave function. 
Indeed, one common indicator of the accuracy of a model is, 
by and large, the percentage of the correlation energy it 
is able to recover. 

For small molecules, variational methods based on 
configuration interaction techniques describe well electronic 
correlations. However, due to their extreme computational cost,
configuration-interaction wavefunctions are noteworthily 
difficult to evaluate for larger systems. Among other procedures
at hand, the correlation can be treated efficiently  by applying a 
Jastrow correlation term to an antisymmetrized wave function 
(e.g.~a single Slater determinant or an antisymmetrized geminal 
power) within quantum Monte Carlo methods \cite{CasulaI,CasulaII}.
The so-called Jastrow antisymmetric geminal ansatz accounts 
for inter-pair interactions and multiple resonance structures, 
maintaining a polynomial scaling cost, comparable to that
of the simpler Jastrow single determinant approach. Highly 
correlated systems, as diradical molecules (the orthogonally 
twisted ethylene C$_2$H$_4$ and the methylene CH$_2$,
for example)
and bond stretching in H$_2$O, C$_2$ and N$_2$, 
are well described by such a method \cite{Zen,Neuscamman}. 

In recent years a 
considerable effort has been devoted to 
characterize the correlation of a quantum system in terms 
of more meaningful quantities, such as the Slater rank for 
two-electron systems \cite{Cirac, Plastino}, the entanglement 
classification for the three-fermion case \cite{Magyares}, the 
squared Frobenius norm of the cumulant part of the two-particle 
reduced density matrix \cite{Mazziotticorr} or the comparison 
with uncorrelated states \cite{PhysRevLett.95.123003}.  
Along with Christian Schilling, we have recently stressed the 
importance of the energy gap in the understanding of the 
electronic correlations \cite{newpaper}. Notwithstanding, 
these measures do not draw a distinction between 
qualitatively different kinds of electronic correlations. In 
quantum chemistry, for instance, it is customary to
distinguish between \textit{static} (or nondynamic) and
\textit{dynamic} correlations. The former corresponds to
configurations which are nearly degenerate with respect to the
reference Slater determinant (if any), whilst the latter arises 
from the need of mixing the Hartree-Fock state with higher-order excited 
states \cite{Becke, Ziesche}. Heuristically, one usually states 
that in systems with (strong) static correlation the wavefunction differs qualitatively from the reference Slater determinant, while strong 
dynamic correlation implies a wavefunction including a large 
number of excited determinants, all with comparable, small 
occupations. Some of the most spectacular failures of the 
Hartree-Fock theory and density functional theory 
(with standard exchange-correlation functionals) are 
related to an incorrect description of static correlation~\cite{Cohen792}. 

It is commonly believed that, to a large extent, both static and 
dynamic contributions should be included in the global computation 
of the electronic correlation. Yet there are few systems for which 
one can distinguish unambiguously between these two types
of correlations. For instance, the ground state of helium has 
no excited electronic states nearby, leading therefore to the 
absence of static correlation. In the dissociation limit of 
H$_2$ a state with fractional occupations arises 
 \cite{Sanchez} and the correlation is purely static. 
According to Hollett and Gill \cite{Hollett}, static correlation 
comes in two ``flavors'': one that can be captured by breaking
the spin symmetry of the Hartree-Fock wave function 
(like in stretched H$_2$) and another that cannot. The 
measures of correlation proposed so far purport to include 
both static and dynamic 
correlations, although in an uncontrolled manner \cite{cumulant}. 

For pure quantum states, global 
structural features of the wave function can be abstracted 
from local information alone. Multiparticle entanglement, for 
instance, can be completely classified with the more accesible 
one-particle picture \cite{Walter1205}. Such a characterization is 
addressed by a finite set of linear inequalities satisfied by 
the eigenvalues of the single-particle states \cite{Sawicki}. 
Furthermore, by using the two-particle density matrix 
and its deviation from idempotency, it is possible to propose a
criterion to distinguish static from dynamic correlation, which for 
two-fermion systems only requires the occupancies of the 
natural orbitals \cite{Matito}. Needless to say, 
grasping global information of a many-body quantum 
system by tackling only one-particle information is quite 
remarkable, mainly because in this way a linear 
number of degrees of freedom is required. 

Recent progress on the N-representability problem  of the 
one-body reduced density matrix for pure states provides
an extension of the well-known Pauli exclusion principle \cite{Kly2}. 
This extension is important because it provides stringent constraints 
beyond those from the Pauli principle, which can be used, 
among others, to improve reduced-density-matrix functional 
theories \cite{RDMFT,recentMazziotti,DePrince}.
Our main aim in this paper is to employ the \textit{generalized Pauli
exclusion principle} to establish a general criterion to distinguish
static and dynamic contributions to the electronic correlation
in fermionic systems. 

The paper is organized as follows. For completeness, 
Section \ref{sec:gpep} summarizes the key aspects of 
the so-called generalized Pauli exclusion principle and 
its potential relevance for quantum chemistry. In Section
 \ref{sec:SL} we discuss a Shull-L\"owdin-type functional 
 for three-fermion systems, which can be constructed by
 using the pertinent generalized Pauli constraints along 
 with the  spin symmetries. Since this functional 
 depends only on the occupation numbers, it is possible to 
 distinguish the correlation degree
 of the so-called Borland-Dennis setting (with an
 underlining six-dimensional one-particle Hilbert space) by 
 using one-particle 
 information alone.  In Section \ref{sec:CM} we discuss 
 a formal way to distinguish static from dynamic correlation.
Section \ref{sec:numerics} is devoted to investigate the 
static and dynamic electronic correlation in molecular 
systems. We compare our results with the well-known 
von-Neumann entanglement entropy. The paper ends 
with a conclusion and an appendix.

\section{The generalization of the Pauli principle}
\label{sec:gpep}

In a groundbreaking work, aimed at solving the quantum 
marginal problem for pure states, Alexander Klyachko 
generalized the Pauli exclusion principle and
provided a set of constraints on the natural 
occupation numbers, stronger than the Pauli 
principle  \cite{Kly2}. Although rudimentary schemes 
to construct such constraints were, to some extent, 
routine in quantum-chemistry literature \cite{Mullern}, 
it was not only with the work of Klyachko that this 
rich structure could be decrypted. 
The main goal of this section is to review the physical 
consequences of such a generalization. 

Given an $N$-fermion state $\ket{\Psi}\in \wedge^N[\H_1]$, 
with $\H_1$ being the one-particle Hilbert space, the \textit{natural
occupation numbers} are the eigenvalues $\{n_{i}\}$ and the 
\textit{natural spin-orbitals} are the eigenvectors $\ket{\varphi_i}$ 
of the one-body reduced density matrix,
\begin{align}
\label{1rdo}
\hat\rho_1\equiv N { \rm Tr}_{N-1}[\ket{\Psi}\langle\Psi|] =
\sum_{i}n_i\ket{\varphi_i}\langle \varphi_i|.
\end{align}
The natural occupation numbers, arranged in 
decreasing order $n_{i} \geq n_{i+1}$, 
 fulfill the Pauli condition $n_1 \leq
1$. The natural spin-orbitals define an orthonormal 
basis $\mathcal{B}_1$ for $\H_1$ and can also be 
used to generate an orthonormal basis 
$\mathcal{B}_N$ for the $N$-fermion Hilbert space 
$\H_N \equiv \wedge^N[\H_1]$, given by the Slater 
determinants
$\ket{\varphi_{i_1}\ldots\varphi_{i_N}} = 
  \ket{\varphi_{i_1}}\wedge \cdots\wedge\ket{\varphi_{i_N}}.$
For practical purposes, the dimension of the one-particle Hilbert
space $\H_1$ is usually finite. Henceforth,
$\H_{N,d}$
denotes an antisymmetric $N$-particle Hilbert space 
with an underlying $d$-dimensional one-particle Hilbert 
space. In principle, the total dimension of
$\H_{N,d}$ is $\binom{d}{N}$, but symmetries 
usually lower it.

It is by now known that the antisymmetry of $N$-fermion pure quantum
states not only implies the well-known Pauli exclusion principle,
which restricts the occupation numbers according to \cite{Col2}
$0\leq n_i\leq 1$, but also entails a set of so-called generalized
Pauli constraints \cite{Kly2,Kly3,CS2013,CSQMath12}. These take the form of 
independent linear inequalities
\begin{equation}
  \label{eq:gpc}
  D_j(\vec{n}) \equiv \kappa_j^{0}+\sum_{i=1}^d\kappa_j^{i}  
  n_i\geq 0 .
\end{equation}
Here the coefficients $\kappa_j^{i} \in \mathbb{Z}$ and
$j=1,2,\ldots,\nu_{N,d}<\infty$. Accordingly, for pure
states the spectrum of a physical fermionic one-body reduced 
density matrix must satisfy a set of independent linear 
inequalities of the type \eqref{eq:gpc}. The total 
number of independent inequalities $\nu_{N,d}$ depends
on the number of fermions and the dimension of the underlying 
one-particle Hilbert space. For instance \cite{Kly3},
$\nu_{3,6} = 4$, $\nu_{3,7} = 4$, $\nu_{3,8} = 31$, 
 $\nu_{3,9} = 52$,  $\nu_{3,10} = 93$, $\nu_{4,8} = 14$,
 $\nu_{4,9} = 60$,
  $\nu_{4,10} = 125$ and $\nu_{5,10} = 161$.

From a geometrical viewpoint, for each fixed pair 
$N$ and  $d$, the family of generalized Pauli
constrains, together with the normalization and the ordering
condition, forms a 
``Paulitope''\footnote{Norbert Mauser
coined the term ``Paulitope'' during the Workshop
\textit{Generalized Pauli Constraints and Fermion Correlation}, 
celebrated at the Wolfgang Pauli Institute in Vienna in August 2016.}: a polytope $\mathcal{P}_{N,d}$ of allowed vectors 
$\vec{n} \equiv (n_i)_{i=1}^d$. 
The
physical relevance of this generalized Pauli exclusion 
principle has been already stressed, among others, in quantum
chemistry~\cite{CSQuasipinning,BenavLiQuasi, Mazz14,
Benavdoubly,
BenavQuasi2, chakraborty2015structure, RDMFT,CSHFZPC, 
CS2016b,TVS16,TVS17,QUA:QUA25376}, 
in open
quantum systems~\cite{MazzOpen} or in condensed
matter~\cite{CS2015Hubbard,CSthesis}. 

The generalized Pauli exclusion principle is particularly 
relevant whenever the natural occupation numbers of 
a given system saturate some of the generalized Pauli constraints~\cite{Kly1}. This so-called ``pinning'' effect
can potentially simplify the complexity of the wave function~\cite{CSHFZPC}. 
In fact, whenever a constraint of the sort \eqref{eq:gpc} is 
saturated or pinned (namely, $D_j(\vec{n}) = 0$), any 
compatible $N$-fermion state $\ket{\Psi}$ (with occupation 
numbers $\vec{n}$) belongs to the null eigenspace of the 
operator
\begin{equation}
  \label{eq:gpcop}
\hat{D}_{j} =  \kappa^{0}_j + \kappa^{1}_j \hat{n}_1 + \cdots
  + \kappa^{d}_j \hat{n}_d,
  \end{equation}
where $\hat{n}_i$ denotes the number operator of
the natural orbital $\ket{\varphi_i}$ of $\ket{\Psi}$.  
This result not only connects the $N$- and 1-particle descriptions,
which is in itself striking, but provides an important selection rule 
for the determinants that can appear in the configuration interaction expansion of the wave function. Indeed, for a given wave function 
$\ket{\Psi}$, 
whenever  $D_j(\vec{n}) = 0$, 
\textit{the Slater determinants for which the relation
$\hat{D}_j\ket{\varphi_{i_1}\ldots\varphi_{i_N}} = 0$ does not 
hold are not permitted in the configuration expansion of $\ket{\Psi}$}. 
In this way, pinned wave functions undergo an extraordinary structural 
simplification which suggests a natural extension 
 of the Hartree-Fock ansatz of the form:
\begin{align}
\ket{\Psi} = \sum_{\{i_1,\ldots, i_N\}\in \mathcal{I}_{D_j}} c_{i_1,\ldots, i_N} \ket{\varphi_{i_1}\ldots\varphi_{i_N}}.
\label{eq:pinningstate}
\end{align}
Here $\mathcal{I}_{D_j}$ stands for the family of configurations 
that may contribute to the wave function in case of pinning to a 
given generalized Pauli constraint $D_j$~\cite{CSHFZPC}. 
These remarkable global implications of extremal local information 
are stable, i.e. they hold approximately for spectra close to the 
boundary of the allowed region \cite{SBV}.

These structural simplifications can be used as a variational 
ansatz, whose computational cost is  cheaper than 
configuration interaction or other post-Hartree-Fock
variational methods \cite{CS2013,CSHFZPC,SBV}. 
For a given hamiltonian $\hat H$, the expectation value 
of the energy $\bra{\Psi}\hat{H}\ket{\Psi}$ 
is minimized with respect to all states $\ket{\Psi}$ of the form  \eqref{eq:pinningstate}, i.e., with natural occupation numbers 
saturating some specific generalized Pauli constraint. For the 
lithium atom, a wave function with three Slater determinants chosen 
in this way accounts for more than 87\% of the total correlation energy
\cite{CSHFZPC}. For harmonium (a system of fermions interacting 
with an external harmonic potential and repelling each other by a 
Hooke-type force), this method accounts for more than 98\% of the 
correlation energy for 3, 4 and 5 fermions \cite{Benavidestesis}.

 In a nutshell, the main aim of the strategy is to select the most 
 important configurations popping up in an efficient configuration  
 interaction computation. We expect that these are the first 
 configurations to appear in approaches whose attempt is also 
 to choose (deterministically \cite{Cleland} or stochastically \cite{Caffarel,CaffarelI}) the most important Slater determinants.
 
\section{The Borland-Dennis setting}
\label{sec:SL}

\subsection{A L\"owdin-Shull functional for three-fermion systems}

The famous Bor\-land-Dennis setting $\H_{3,6}$, 
the rank-six approximation for the three-electron system,
is completely characterized by 4 constrains~\cite{Borl1972}: 
the equalities 
\begin{align}
n_{1} + n_{6} = n_{2} + n_{5} = 
n_{3} + n_{4} = 1
  \label{eq:BD1}
\end{align}
 and the inequality:
\begin{align}
  n_{1} + n_{2} + n_{4} \leq 2 . 
  \label{eq:BD2}
\end{align}
This latter inequality together with the decreasing
ordering rule defines a polytope in $\R^6$, called here the 
Borland-Dennis Paulitope. 
Conditions   \eqref{eq:BD1} imply that, in the natural orbital 
basis, every Slater determinant, built up from three natural 
spin-orbitals, showing up in the configuration 
expansion \eqref{eq:pinningstate}, satisfies  
\begin{align}
\ket{\varphi_i\varphi_j\varphi_k} =
(\hat n_{7-s} + \hat n_{s})\ket{\varphi_i\varphi_j\varphi_k},
\end{align}
for $s\in\{1,2,3\}$.
Therefore, each natural spin-orbital belongs 
to one of three different sets, say
$\varphi_i \in \{\varphi_1,\varphi_6\}$, 
$\varphi_j \in \{\varphi_2,\varphi_5\}$ and
$\varphi_k \in \{\varphi_3,\varphi_4\}$. 
Consequently, the dimension of the total Hilbert 
space is eight.

In the symmetry-adapted description of this system, the spin of 
three natural orbitals points down, and the spin of the
other three points up. The corresponding one-body 
reduced density matrix (a $6\times 6$ matrix) is
a block-diagonal matrix that can be written as the
direct sum of two ($3\times 3$) matrices (say,
$\hat \rho_\up$ and $\hat \rho_\dn$),
one related to the spin up and the other one related to 
the spin down. For the doublet configuration,
each acceptable Slater determinant contains two spin
orbitals pointing up (for instance) and one pointing down. 
It follows that
\begin{equation}
\Tr \hat  \rho_\up = 2 \word{and} \Tr \hat \rho_\dn = 1.
\end{equation}
To meet the decreasing ordering of the natural occupations
as well as the representability conditions \eqref{eq:BD1}, 
two of the first three occupation numbers must belong to the 
matrix whose trace is equal to two \cite{BenavQuasi2}.
It is straightforward to see that the only admisible 
set of occupation numbers are the ones lying  
in the hyperplane $\mathcal{A}_1$:
\begin{align}
\label{eq:const1}
n_1 + n_2 + n_4 = 2
\end{align}
(equivalently, $n_3 + n_5 + n_ 6= 1$), 
which saturates the generalized Pauli  
constraint~\eqref{eq:BD2}, or in the hyperplane 
$\mathcal{A}_2$:
 \begin{align}
\label{eq:const2}
n_1 + n_2 + n_3 = 2
\end{align}
(equivalently, $n_4 + n_5 + n_6 = 1$).
Note that the two hyperplanes intersect on the line
$n_3 = n_4 = \tfrac12$. In Fig.~\ref{graf:polytope} 
hyperplanes $\mathcal{A}_1$ and $\mathcal{A}_2$
are shown within the Pauli hypercube $n_1 \leq 1$.

As stated above, pinning of natural occupation numbers 
undergoes a remarkable structural simplification of the
wave functions compatible with $\vec{n}$. 
For the case of the Borland-Dennis setting,
the equation \eqref{eq:const1} implies for the corresponding
wave function the condition
$(\hat{n}_{1} + \hat{n}_{2} + \hat{n}_{4})\ket{\Psi} = 2 \ket{\Psi}$, 
while the constraint \eqref{eq:const2} implies
$(\hat{n}_{1} + \hat{n}_{2} + \hat{n}_{3})\ket{\Psi} = 2 \ket{\Psi}$.
Consequently, a wave function (the so-called Borland-Dennis state)
compatible with the hyperplane 
$\mathcal{A}_1$ can be written in the form \cite{BenavQuasi2}:
\begin{align}
\label{eq:BDstateoriginal}
\ket{\Psi_{\rm BD}} = \sqrt{n_3} \, \ket{\varphi_1\varphi_2\varphi_3} + \sqrt{n_5} \,
\ket{\varphi_1\varphi_4\varphi_5} + \sqrt{n_6} \, \ket{\varphi_2\varphi_4\varphi_6},
\end{align}
where $n_{3} \geq n_{5} + n_{6}$ and $n_{3} \geq \tfrac12$.
A wave function 
compatible with the hyperplane $\mathcal{A}_2$ reads:
\begin{align}
\label{eq:BDstateoriginalII}
\ket{\Psi_2} = \sqrt{n_4} \, \ket{\varphi_1\varphi_2\varphi_4} + \sqrt{n_5} \,
\ket{\varphi_1\varphi_3\varphi_5} + \sqrt{n_6} \, \ket{\varphi_2\varphi_3\varphi_6},
\end{align} 
where $n_{4} \leq n_{5} + n_{6}$ and $n_{4} \leq \tfrac12$.
Notice that, just like in the famous L\"owdin-Shull functional for 
two-fermion systems~\cite{LS}, the wave function is 
explicitly written in terms of both the natural occupation 
numbers and the natural orbitals. Likewise, any sign dilemma 
that may occur when writing the amplitudes of the states
\eqref{eq:BDstateoriginal} and \eqref{eq:BDstateoriginalII}
can be dodged by absorbing the phase into the spin-orbitals.
Moreover, only doubly excited configurations are 
permitted here.  For  $\ket{\Psi_{\rm BD}}$ such double 
excitations are referred to the Slater determinant whose 
one-particle density matrix is the best idempotent 
approximation to the true one-particle density matrix.
The state $\ket{\Psi_2}$ is orthogonal to  the state
$\ket{\varphi_1\varphi_2\varphi_3}$. Interestingly, a
non-vanishing overlap of a wave function with this 
latter state can only be guaranteed if the sum of the 
first three natural occupation numbers 
is larger than two~\cite{KS68}.  Both $\ket{\Psi_{\rm BD}}$ 
and $\ket{\Psi_2}$ lead to diagonal one-particle reduced 
density matrices.
 
 For any given Slater determinant, the seniority number
 is defined as the number of orbitals which are singly occupied. 
 Such an important concept is used in nuclear and condensed
 matter physics to partition the Hilbert space and construct 
 compact configuration-interaction wave functions
 \cite{Seniority,SeniorityI}. Since the wave functions 
 \eqref{eq:BDstateoriginal} and \eqref{eq:BDstateoriginalII}
 are eigenfunctions of the spin operators, each Slater determinant
 showing up in these expansions is also an eigenfunction of 
 such operators. The latter is only possible if one orbital is 
 doubly occupied and therefore the seniority number of each 
 Slater determinant is 1. The seniority number of 
 $\ket{\Psi_{\rm BD}}$ and $\ket{\Psi_2}$ is also 1.

\begin{figure}[!t] 
 \centering
\includegraphics[width=8.5cm]{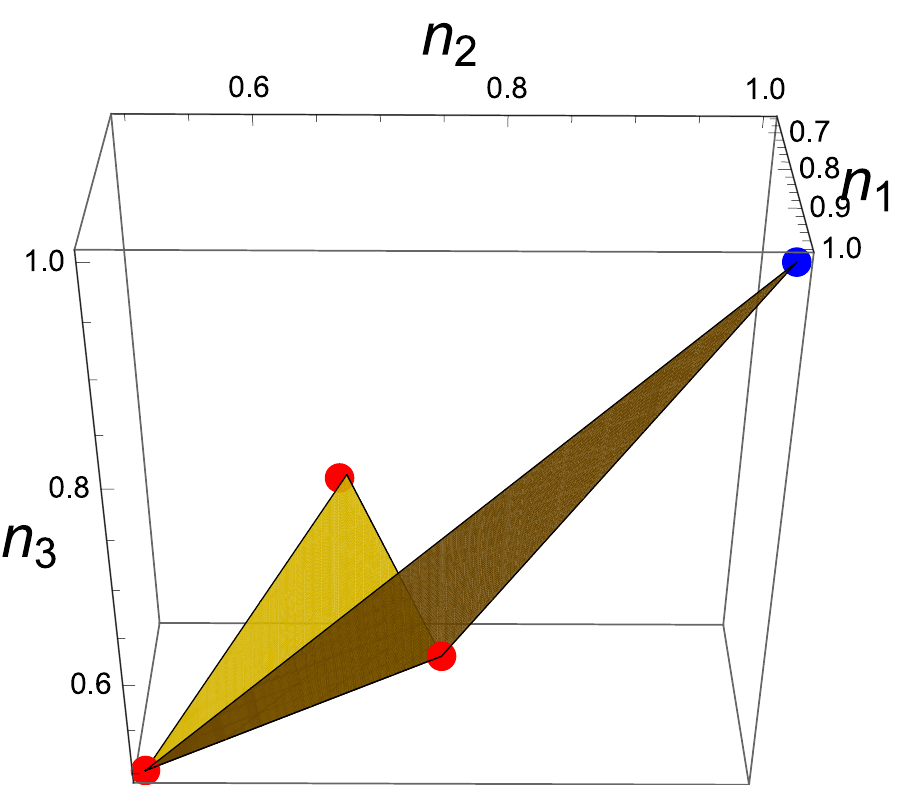}
\caption{The hyperplanes $2 = n_1 + n_2 + n_4$ 
and $2 = n_1 + n_2 + n_3$, subject to the conditions 
$1 \ge n_1 \ge n_2 \ge n_3 \geq 0.5$ and 
$2 \geq n_1 + n_2 + n_4$. 
 The blue dot is the Hartree-Fock point $\bigl(1,1,1\bigr)$. 
 The red ones are $\bigl(1,\tfrac12,\tfrac12\bigr)$, 
 $\bigl(\tfrac34,\tfrac34,\tfrac12\bigr)$ and 
 $\bigl(\tfrac23,\tfrac23,\tfrac23\bigr)$.}
 \label{graf:polytope}
 \end{figure}

\subsection{Correlations}

In Fig.~\ref{graf:polytope} we illustrate the discussion of the 
previous section, highlighting four special 
configurations, namely:

\begin{itemize}

\item The ``Hartree-Fock'' point 
$(n_1,n_2,n_3) = \bigl(1,1,1\bigr)$, which 
corresponds to the single Slater determinant
$\ket{\varphi_1\varphi_2\varphi_3}$.
Note that it does not coincide in general with
the Hartree-Fock state, since it is described in the 
natural-orbital basis set. However, we call it so 
because its spectrum is
$\vec{n}_{\rm HF} = (1,1,1,0,0,0)$.

\item The point 
$\mathcal{P}_a \equiv  \bigl(\tfrac23,\tfrac23,\tfrac23\bigr)$,
which corresponds to the strongly (static) correlated state: 
\begin{align}
\label{eq:a}
\ket{\Psi_a} =
\tfrac1{\sqrt{3}}  (\ket{\varphi_1\varphi_2\varphi_4} 
+ \ket{\varphi_1\varphi_3\varphi_5} 
+ \ket{\varphi_2\varphi_3\varphi_6}).
\end{align} 

\item The point $\mathcal{P}_b \equiv   \bigl(1,\tfrac12,\tfrac12\bigr)$.
These occupation numbers correspond to the state
\begin{align}
\label{eq:b}
\ket{\Psi_b} = \tfrac1{\sqrt{2}} (\ket{\varphi_1\varphi_2\varphi_3} 
+ \ket{\varphi_1\varphi_4\varphi_5}).
\end{align} 
In quantum information theory, this state is said to be 
 biseparable because one of the particles is disentangled 
 from the other ones \cite{Vrana2008}. 

\item The point $\mathcal{P}_c \equiv  
\bigl(\tfrac34,\tfrac34,\tfrac12\bigr)$,
which correspond to the (static) correlated state:
\begin{align}
\label{eq:c}
\ket{\Psi_c} = 
\tfrac1{\sqrt{2}} \ket{\varphi_1\varphi_2\varphi_3} +
\tfrac12
 (\ket{\varphi_1\varphi_4\varphi_5} 
+ \ket{\varphi_2\varphi_4\varphi_6}).
\end{align} 
Points $\mathcal{P}_b$ and $\mathcal{P}_c$ lie in the 
intersection of $\mathcal{A}_1$ and $\mathcal{A}_2$,
namely, the degeneracy line $n_3 = n_4 = \half$. 
Since $n_3$ and $n_4$ are identical, the choice of the 
highest occupied natural orbital $\ket{\varphi_3}$ and the 
lowest unoccupied natural orbital $\ket{\varphi_4}$ is not 
unique anymore and the indices $3$ and $4$ can be 
swapped in \eqref{eq:b} and \eqref{eq:c} without
changing the spectra.

\end{itemize}
 
These four points are important because they belong to two 
different correlation regimes.  On the one hand, the states 
$\ket{\Psi_a}$ and $\ket{\Psi_b}$ exhibit static correlation, 
as they are equiponderant superpositions of Slater 
determinants. On the other, the state $\ket{\Psi_c}$
is the superposition of two states
($ \ket{\varphi_1\varphi_4\varphi_5}$ and
$\ket{\varphi_2\varphi_4\varphi_6}$) and the nearly
degenerate $\ket{\varphi_1\varphi_4\varphi_5}$ and 
its correlation is also static. 
This is reminiscent of the zero-order description of the 
beryllium ground state, for which the $2s$ and $2p$ orbitals 
are nearly degenerate and the state is a equiponderant 
superposition of three Slater determinants plus a highly 
weighted reference state  \cite{TewKlopperHelgaker}.

\begin{figure}[ht] 
 \centering
\includegraphics[width=7cm]{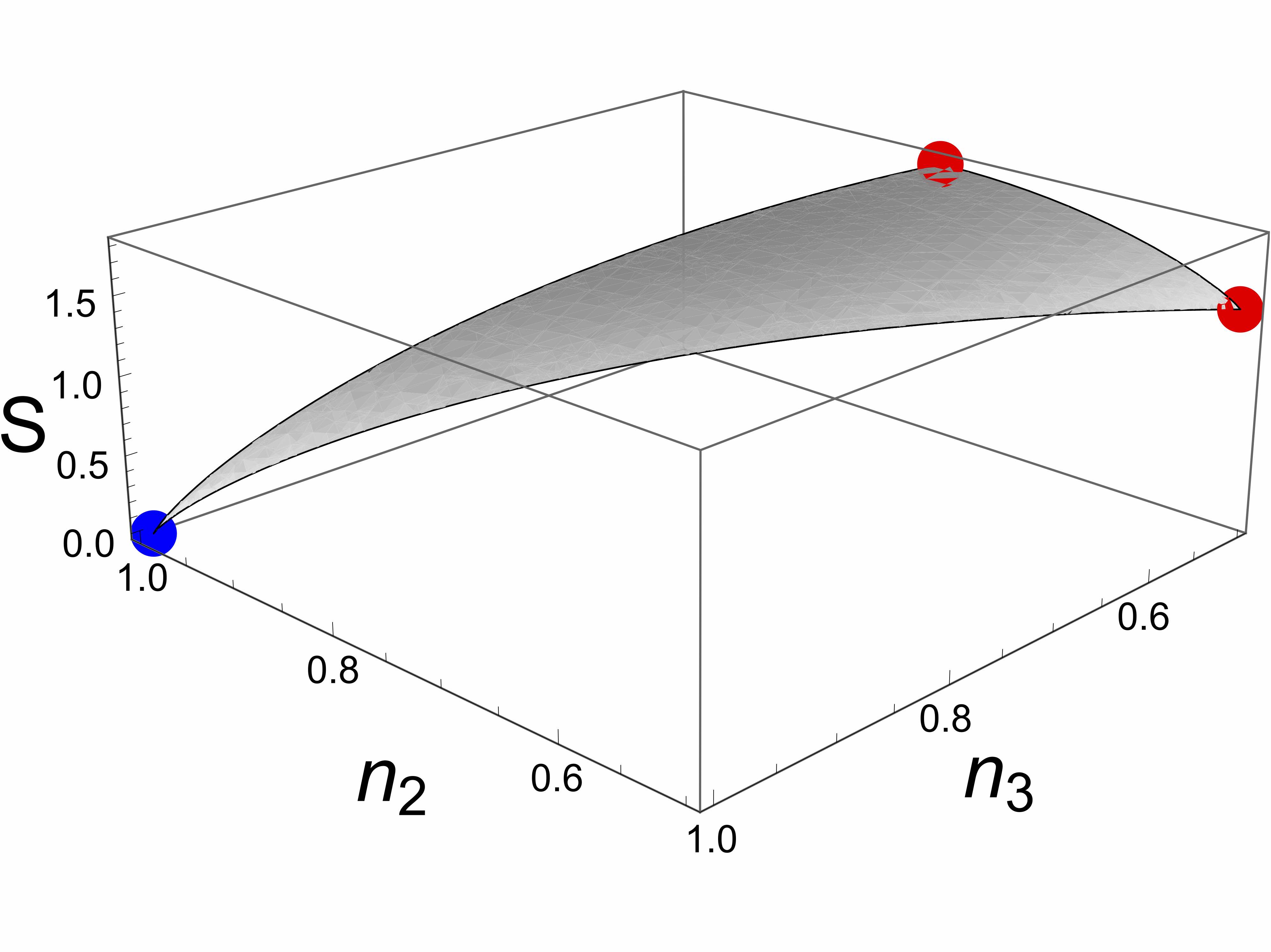}
\includegraphics[width=7cm]{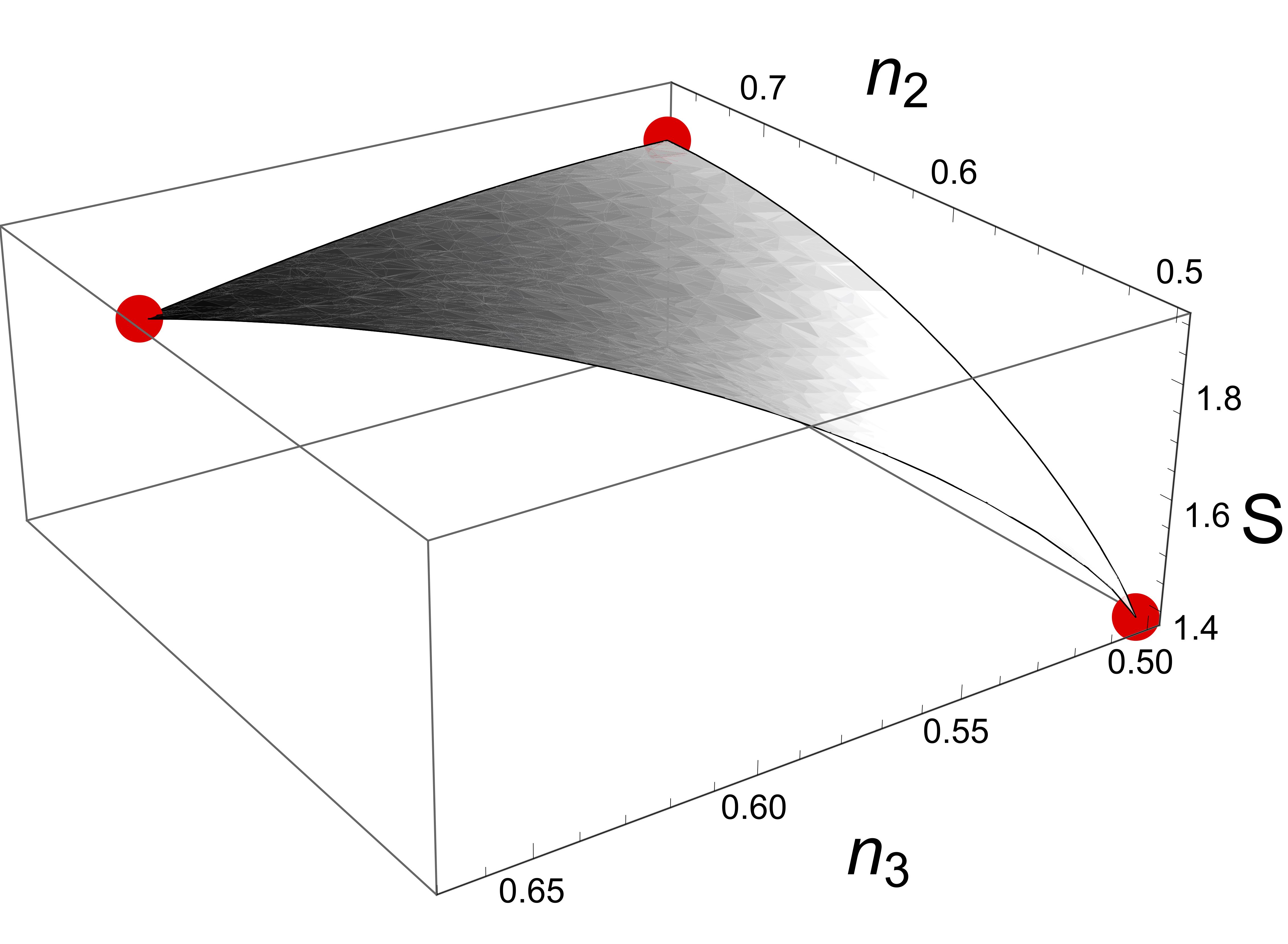}
\caption{Entanglement entropy of the hyperplanes $\mathcal{A}_1$
and $\mathcal{A}_2$.}
 \label{graf:VN}
 \end{figure}
 
In Fig.~\ref{graf:VN} the entanglement entropy $S = -\sum_i 
n_i \ln n_i$ is plotted as a function of $n_2$ and $n_3$ for 
the hyperplanes $\mathcal{A}_1$ and $\mathcal{A}_2$. 
The entanglement entropy of the state 
$\ket{\varphi_1\varphi_2\varphi_3}$ is zero since it is 
uncorrelated.  The entropies
of the states $\ket{\Psi_b}$ and $\ket{\Psi_c}$
are 1.3862 and 1.8178, respectively. As one might expect,
the configurations present in $\mathcal{A}_2$ all are 
strongly correlated. For the highest correlated state 
$\ket{\Psi_a}$, $S = 1.9095$. 
The particular structure of the states $\ket{\Psi_a}$, $\ket{\Psi_b}$
and $\ket{\Psi_c}$ prompts us to say that the correlation effects 
of the states lying in the hyperplane $\mathcal{A}_2$ are all due 
to static effects, while the states in $\mathcal{A}_1$ are due to
both static and dynamic effects.

According to the particle-hole symmetry, when 
applied to a three-electron system, the nonzero
eigenvalues and their multiplicities are the 
same for the one- and the two-body reduced 
matrices. Thus, the results in this section based 
on one-particle information alone are also valid 
at the level of the 2-particle picture.

 \subsection{Borland-Dennis for three active electrons}
 
 In the configuration interaction picture, the full wavefunction
is to be expressed in a given one-electron basis as a linear 
combination of all possible Slater determinants, save symmetries. 
In the basis of natural orbitals, it reads:
\begin{align}
  \label{eq:CI}
  \ket{\Psi} = \sum_{1\leq i_1< \cdots < i_N \leq d} c_{i_1\ldots i_N}  \ket{\varphi_{i_1}\ldots\varphi_{i_N}}
\end{align}
in a similar fashion to the Hartree-Fock ansatz
\eqref{eq:pinningstate}.  It is well known that the 
expansion~\eqref{eq:CI} contains a very large 
number of configurations that are superfluous or 
negligible for computing molecular electronic properties. 
In practice, the configurations considered effective 
are sparse if an arbitrary threshold for the value of the 
amplitudes in \eqref{eq:CI} is enforced~\cite{Mentel2014}. 
As such, one often introduces the notion
of active space to select the most relevant configurations at the
level of the one-particle picture. A \textit{complete active space}
classifies the one-particle Hilbert space in core (fully occupied), active
(partially occupied) and virtual (empty) spin-orbitals.
The core spin-orbitals are pinned (completely populated) 
and are not treated as correlated. 
Adding active-space constraints
 improves the estimate of the 
ground-state energy in the framework of reduced-density-matrix
theory \cite{acrdm}.

The generalized Pauli principle can shed some light on this 
important concept \cite{TVS17,SBV}. In fact, 
for the case of $r$ core (and consequently $d-r$ active orbitals) 
the Hilbert space $\H_{N,d}$ is isomorphic to the wedge product 
$\H^{\rm core}_{r,r} \wedge \H^{\rm active}_{N-r,d-r}$.
Hence, a wave function 
 $\ket{\Psi} \in   \H_{N,d}$ can be written in the following way:
 \begin{equation}
 \ket{\Psi} = \ket{\varphi_1\dots\varphi_r}  \wedge \ket{\Psi^{\rm active}},
 \end{equation}
 where $\ket{\Psi^{\rm active}} \in \H^{\rm active}_{N-r,d-r}$. 
The first $r$ natural occupation numbers 
are saturated to $1$. The remaining $d-r$ occupation numbers 
$(n_{r +1}, \dots, n_d)$
satisfy a set of generalized Pauli constraints
and lie therefore inside the polytope $\mathcal{P}_{N-r,d-r}$. 
The space $\H^{\rm active}_{N-r,d-r}$ is called here the 
``active Hilbert space''. 
For instance, for the ``Hartree-Fock'' space
$\H_{N,N}$, the corresponding zero dimensional active 
Hilbert space is $\H^{\rm active}_{0,0}$.

It is possible to characterize a hierarchy of active 
spaces by the effective dimension of 
$\H^{\rm active}_{N-r,d-r}$ and the number 
of Slater determinants appearing in the configuration 
interaction expansion of $\ket{\Psi^{\rm active}}$~\cite{TVS17,SBV}. 
For the ``active'' Borland-Dennis setting $\H^{\rm active}_{3,6}$ 
we can apply the same considerations discussed in the last subsections:
if the corresponding constraint \eqref{eq:BD2}
is saturated, the wave function fulfills
$(\hat{n}_{r+1} + \hat{n}_{r+2} + \hat{n}_{r+4})\ket{\Psi} = 2 \ket{\Psi}$, 
and the set of possible Slater determinants reduces to just three,
taking thus the form:
\begin{align}
  \ket{\Psi^{\rm active}} &= \sqrt{n_{r+3}} \ket{\varphi_{r+1}\varphi_{r+2}\varphi_{r+3}}+ 
  \sqrt{n_{r+5}}  \ket{\varphi_{r+1}\varphi_{r+4}\varphi_{r+5}} 
  \nonumber
  \\ &+
 \sqrt{n_{r+6}} \ket{\varphi_{r+2}\varphi_{r+4}\varphi_{r+6}},
  \label{eq:BDstate}
\end{align}
provided that $n_{r+3} \geq \tfrac12$ and 
$n_{r+3} \geq n_{r+5} + n_{r+6}$.

\section{Correlations and correlation measures}
\label{sec:CM}

Even if the peculiar role played by electronic correlations 
 in quantum mechanics were noticed from the onset,
 the problem of how to measure quantum correlations 
 is still subject to an intense research \cite{Horo,PhysRevA.92.042329}. The degree of entanglement $\mathcal{D}(\Psi)$ of an arbitrary vector 
 $\ket{\Psi}$ can be expressed by its projection onto 
 the nearest normalized unentangled (or uncorrelated)
 pure state \cite{Shimony,Myers2010}:
  \begin{align}
\mathcal{D}(\Psi) = 1 - \max_{\Phi} |\langle\Psi \ket{\Phi}|^2,
  \label{eq:Myers}
  \end{align}
  where the maximum is over all unentangled states, 
  normalized so that $\langle\Phi\ket{\Phi} = 1$.
 Although this measures sound conventional, 
 it has the merit of being zero whenever
  $\ket{\Psi}$ is uncorrelated. 
  
This measure (and the minimum $\min_{\tilde\Phi\in\mathcal{F}} 
|| \Psi - \tilde \Phi||^2$, where $\mathcal{F}$ denotes the set
of unnormalized unentangled (or uncorrelated) pure 
states\cite{Shimony}) is also important in the realm of 
quantum chemistry. 
In the Appendix we state and prove that
the set of pure quantum systems with 
predetermined energy is connected: given a 
Hamiltonian $\hat{H}$ there are two wavefunctions 
$\ket{\psi_1}$ and $\ket{\psi_2}$ with energies $E_1$
and $E_2$ whose distance $|| \psi_1 - \psi_2||^2$ is 
bounded by a function of $|E_1 - E_2|$
(see Theorem~\ref{thm2} in the Appendix).   
Recently, we have shown that when $\ket{\Psi}$ is the 
  ground state of a given Hamiltonian the measure 
  \eqref{eq:Myers} is closely related to the concept of 
  correlation energy as understood in 
  quantum chemistry \cite{newpaper}.
   A key ingredient in such connection turns out to be 
 the energy gap within the symmetry-adapted Hilbert subspace. 
  
  \subsection{Dynamic correlation}
  
The $N$-particle description of a quantum system 
 and its reduced one-fermion picture can be
related in meaningful ways. In effect, 
$\mathcal{D}(\Psi)$
can be bounded from above and from below by the 
$l^1$-distance of the natural occupation numbers. In fact, the 
distance between a wave function  $\ket{\Psi}$ and any 
Slater determinant $\ket{\varphi_{i_1}\dots\varphi_{i_N}}$
satisfies \cite{CS2013}:
\begin{align}
\frac{\delta_{\ii}(\vec n)}{2 \min (N,d-N) } \leq 1 - | 
\bra{\varphi_{i_1}\dots\varphi_{i_N}} \Psi\rangle|^2 \leq 
\frac{\delta_{\ii}(\vec n)}2,
\end{align}
where $d$ is the dimension of the underlying one-particle Hilbert space,
as defined in Sec.~\ref{sec:gpep},
and $\delta_{\ii}(\vec n) \equiv \sum_{i\in \ii} (1 - n_i) + 
\sum_{i \notin \ii} n_i$ is the $l^1$-distance between $\vec n$
(the natural occupation numbers of $\ket{\Psi}$) and
the natural occupation numbers of the Slater determinant
in display (here $\ii \equiv \{i_1,\dots,i_N\}$). This 
result is also valid for Hartree-Fock or Brueckner 
orbitals~\cite{newpaper,Zhan}. In particular, 
the $l^1$-distance to the Hartree-Fock point is given by:
\begin{align}
  \label{dynamicalcorr}
  \delta_{\rm HF}(\vec n) = \sum_{i \leq N} (1 - n_i) + 
\sum_{i > N} n_i.
\end{align}
In Fig.~\ref{graf:DHF} we plot $\delta_{\rm HF}(\vec n)/2$
for the points in the hyperplane $\mathcal{A}_1$. As 
expected, $\delta_{\rm HF}(\vec n_{\rm HF}) = 0$. More 
interesting, the correlation increases monotonically
with $n_3$. All the points on the degeneracy line $n_3 = n_4$
are at the same $l^1$-distance from the Hartree-Fock point. 
In effect, 
$\delta_{\rm HF}(\vec n_{\rm s}(\eta))/2 = 1$,
where 
\begin{align}
\label{eq:staticstate}
\vec n_{\rm s}(\eta) \equiv (\tfrac34 + \eta, \tfrac34 - \eta,
\tfrac12, \tfrac12, \tfrac14 + \eta, \tfrac14 - \eta),
\end{align}
with $0 \leq \eta \leq \tfrac14$, is the set of points lying on 
the intersection line $n_3 = n_4 = \tfrac12$.
Moreover, 
$\delta_{\rm HF}(\mathcal{A}_2)/2 = 1$, for all the points 
lying on the
hyperplane $\mathcal{A}_2$.

\begin{figure}[ht] 
 \centering
\includegraphics[width=8.5cm]{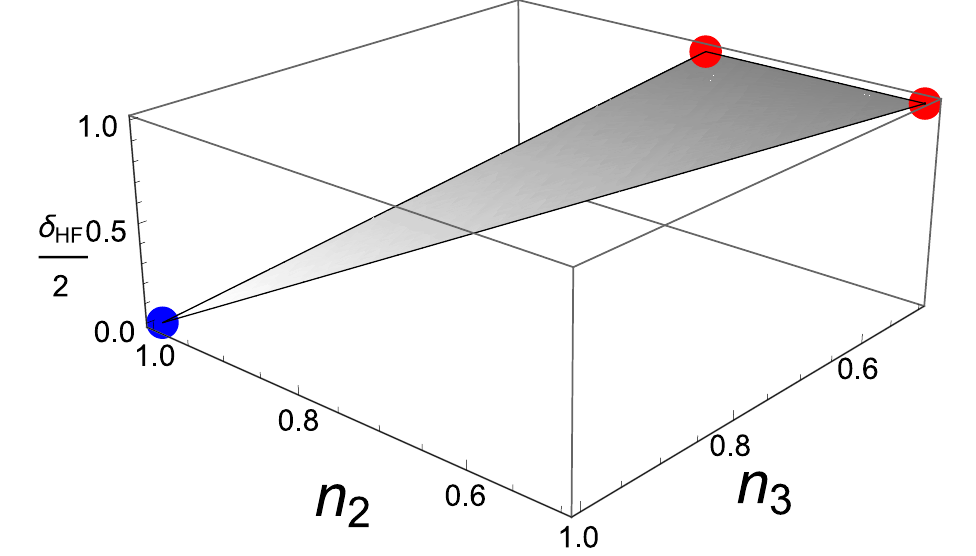}
\caption{$l^1$-distance of the hyperplane $\mathcal{A}_1$ 
with respect to the Hartree-Fock point $\vec n_{\rm HF}$.}
 \label{graf:DHF}
 \end{figure}
 
   \subsection{Static correlation}
   
Roughly speaking, the idea of static correlation 
is associated with the presence of a wave function built up 
from an equiponderant superposition of more than one 
Slater determinant, namely,
\begin{align}
\label{eq:equiponder}
\frac1{\sqrt{m}}(\ket{\Phi_1} + \cdots + \ket{\Phi_m}),
\end{align}
For the Borland-Dennis setting $\H_{3,6}$, the hyperplane 
$\mathcal{A}_2$ contains states with three
configurations being almost equiponderant. 
After this long discussion, it is natural to
define all the points lying on the hyperplane $\mathcal{A}_2$
as statically correlated.
The hyperplane $\mathcal{A}_1$ contains the uncorrelated 
Hartree-Fock state and the correlation of the rest of 
the states present is due to static as well as dynamic effects.

The ``static'' states $\ket{\Psi_{\rm s}(\eta)}$ that lead to the
occupancies $\vec n_{\rm s}(\eta)$ as defined in Eq.~\eqref{eq:staticstate}
read:
\begin{equation}
\ket{\Psi_{\rm s}(\eta)} \equiv
\tfrac1{\sqrt{2}} \ket{\varphi_1\varphi_2\varphi_3} 
+ \sqrt{\tfrac14 + \eta}\ket{\varphi_1\varphi_4\varphi_5} +
\sqrt{\tfrac14 - \eta}\ket{\varphi_2\varphi_4\varphi_6},
\end{equation}
where $0 \leq \eta \leq \tfrac14$. As stated before, since 
$n_3$ and $n_4$ are identical, the choice of the highest 
occupied natural orbital and the lowest unoccupied natural 
orbital is not unique and the indices 3 and 4 can be swapped 
without changing the spectra. However, by doing so the resulting
state is orthogonal to $\ket{\varphi_1\varphi_2\varphi_3}$.

The $L^2$-distance 
(Eq.~\eqref{eq:Myers})
between the Borland-Dennis state 
\eqref{eq:BDstateoriginal} and $\ket{\Psi_{\rm s}(\eta)}$ is 
given by $\mathcal{J}_{\rm s}(\eta) \equiv 1 - 
|\langle\Psi_{\rm BD}\ket{\Psi_{\rm s}(\eta)}|^2$. 
The minimum of this distance depends only on the value of 
the natural occupation number corresponding to the highest
occupied natural orbital. To see this notice that the minimum 
of $\mathcal{J}_{\rm s}(\eta)$ is attained when 
$d\mathcal{J}_{\rm s}(\eta)/d\eta = 0$, which happens 
at
$
\eta^{*} = \frac{n_5 - n_6}{4(1 - n_3)}.
$
The distance is therefore
\begin{align}
\mathcal{J}_{\rm s}(\eta^{*}) &= 1 - 
\bigg(\frac{\sqrt{n_3}}{\sqrt{2}} + \frac{1-n_2}{\sqrt{2(1-n_3)}}
+ \frac{n_2-n_3}{\sqrt{2(1-n_3)}}\bigg)^2 
\nonumber \\
&= 1 - \tfrac12(\sqrt{n_3} + \sqrt{1-n_3})^2 = \tfrac12 - \sqrt{n_3(1-n_3)},
\end{align}
which is zero when the correlation of the 
Borland-Dennis state is completely static and 
is $\tfrac12$ when the state is the uncorrelated
Hartree-Fock state. 

We can also investigate the $l^1$-distance between 
any state $\vec n \in \mathcal{A}_1$ and the state
$\vec n_{\rm s}(\eta)$ \eqref{eq:staticstate}, namely:
\begin{align}
  \label{staticalcorr}
  \delta_{\rm s}(\vec n) = \min_{\eta} dist_1(\vec n, \vec n_s(\eta)).
\end{align}
Notice that the $l^1$-distance reads
$$
dist_1(\vec n, \vec n_s(\eta)) = 2(|n_1 - (\tfrac34 + \eta)|
+  |n_2 - (\tfrac34 - \eta)| +  |n_3 - \tfrac12|).
$$
So, to minimize $dist_1(\vec n, \vec n_s(\eta))$
 is the same as minimizing $|n_1 -\tfrac34 - \eta|
+  |n_2 - \tfrac34 + \eta|$.
Since the $l^1$-sphere with radius $z$
centered at $\vec n$ is the convex hull of the vertices
$((n_1\pm z,n_2,n_3),(n_1,n_2\pm z,n_3),
(n_1,n_2,n_3\pm z))$, the minimum of the distance
is:
\begin{align}
\label{eq:minimuml}
\delta_{\rm s}(\vec n) &= 
 2(|n_1 - (-n_2+\tfrac32)| +  |n_3 - \tfrac12|) 
% =  2(|1 + n_3 - \tfrac32)| +  |n_3 - \tfrac12|) 
 =  4(n_3 - \tfrac12),
\end{align}
which depends on $n_3$ alone. Notice that for 
the Hartree-Fock point $\delta_{\rm s}(\vec n_{\rm HF})/2 = 1$
and remember that by our definition 
$\delta_{\rm s}(\vec n) = 0$ if $\vec n \in \mathcal{A}_2$. 
Remarkably, $\eta^{*}$, the minimizer of  $\mathcal{J}_{\rm s}(\eta)$, 
is also a minimizer of $dist_1(\vec n, \vec n_s(\eta))$.
This latter statement can be proved by noting 
that the first occupation number of $\vec n_s(\eta_{*})$ 
satisfies 
$$
n_1(\eta^{*}) = \frac34 +  \frac{n_1 - n_2}{4(1 - n_3)}
= \frac{2 - 4n_3 + 2n_1}{4(1 - n_3)} \leq \frac{n_1}{2(1 - n_3)}
\leq n_1,
$$
since $n_3 \geq \tfrac12$. 
Equivalently, $n_2(\eta^{*}) \leq n_2$.
Therefore:
\begin{align*}
dist_1(\vec n, \vec n_s(\eta^{*})) &= 2[(n_1 - \tfrac34 -  \eta_{*})
+  (n_2 - \tfrac34  + \eta^{*}) + (n_3- \tfrac12)] \\
&= 2[(n_1 + n_2 -  \tfrac32) + (n_3- \tfrac12)] =
4 (n_3- \tfrac12),
\end{align*}
which is the minimum \eqref{eq:minimuml}.

  \subsection{Correlation measures}
   
The comparison of  the distances \eqref{dynamicalcorr} and
\eqref{staticalcorr} allows us to distinguish between dynamic 
and static correlations. The distance to the Hartree-Fock point
$\delta_{\rm HF}(\vec n)$
can be viewed as a measure of the dynamic part of the 
electronic correlation, for it quantifies how much a wave 
function differs from the uncorrelated 
Hartree-Fock state. The static distance $\delta_{\rm s}(\vec n)$ 
can be viewed as a measure of the 
static part of the correlation, as it quantifies how much a wave 
function differs from the set of static states. One expects for helium
$\delta_{\rm HF} (\vec n_{\rm He}) \approx 0$ while for 
H$_2$ at infinite separation 
$\delta_{\rm s} (\vec n_{{\rm H}_{2,\infty}}) = 0$.
For convenience we renormalize these two $l^1$-distances 
by means of
\begin{align}
  P_{\rm sta}(\vec n)  = 
  \frac{\delta_{\rm HF}(\vec n)}{\delta_{\rm HF}(\vec n)  
  +\delta_{\rm s}(\vec n)} 
\quad {\rm and}
 \quad  P_{\rm dyn}(\vec n) = \frac{\delta_{\rm s}(\vec n)}
  {\delta_{\rm HF}(\vec n) 
  +\delta_{\rm s}(\vec n)}.
\end{align}
Since the measures $P_{\rm dyn}(\vec n)$ and 
$P_{\rm sta}(\vec n)$ are normalized (while 
$\delta_{\rm HF}(\vec n)$ and $\delta_{\rm s}(\vec n)$ 
are not), they are much
more useful to compare different systems. 
In this way, when the correlation is due to static effects 
$P_{\rm sta}(\vec n)  = 1$ and $P_{\rm dyn}(\vec n)  = 0$,
while the contrary occurs for a completely dynamic state.
For instance, one expects for He 
$P_{\rm sta} (\vec n_{\rm He}) \approx 0$ 
while for H$_2$ at infinite separation 
$P_{\rm dyn} (\vec n_{{\rm H}_{2,\infty}}) = 0$.
These quantities have the merit of being zero or one when 
the correlation is completely dynamic or completely static
and hence they separate the correlation in two contributions.

It is worth saying that our considerations do not only 
apply for the Borland-Dennis setting but also for larger 
ones.  Recall that for the settings 
$\mathcal{H}_{N,d}$ and $\mathcal{H}_{N,d'}$, such that
$d < d' \in \N$, the corresponding polytopes satisfy:
$\mathcal{P}_{N,d} = {\mathcal{P}_{N,d'}}|
_{n_{d+1} = \cdots = n_{d'} = 0}$.
It means that, intersected with the hyperplane 
given by $n_{d+1} = \cdots = n_{d'} = 0$, the 
polytope $\mathcal{P}_{N,d'}$ coincides with 
$\mathcal{P}_{N,d}$ \cite{CS2013}. Therefore, 
we have completely characterized the static states up
to six dimensional one-particle Hilbert spaces. By choosing 
$n_1 =1$, we freeze one electron, we are effectively
dealing with a two active-electron system and our measures 
can also be used for characterizing the correlation of 
two-electron systems. Note that for this latter case the 
natural occupation numbers are evenly degenerated,
a very well known representability condition for systems 
with a even number of electrons and time-reversal symmetry 
\cite{SmithDarwin}.

\section{Correlation in molecular systems}
\label{sec:numerics}

To illustrate these concepts we plot in Fig.~\ref{graf:H2}
our measures of static and dynamic correlation for the ground 
states of the diatomic molecules 
H$_2$ and Li$_2$ as a function of the interatomic distance. 
In the same plot we can also see the von-Neumann entropy 
and the value of the highest occupancy of the highest 
occupied natural spin-orbital. These numbers were obtained 
from CAS-SCF calculation using the code Gamess \cite{Gamess} 
and cc-pVTZ basis sets with all electrons active and as large 
as possible active space of orbitals.
 
 \begin{figure}[!t] 
\centering
      \subcaptionbox{H$_2$.}
        {\includegraphics[width=8cm]{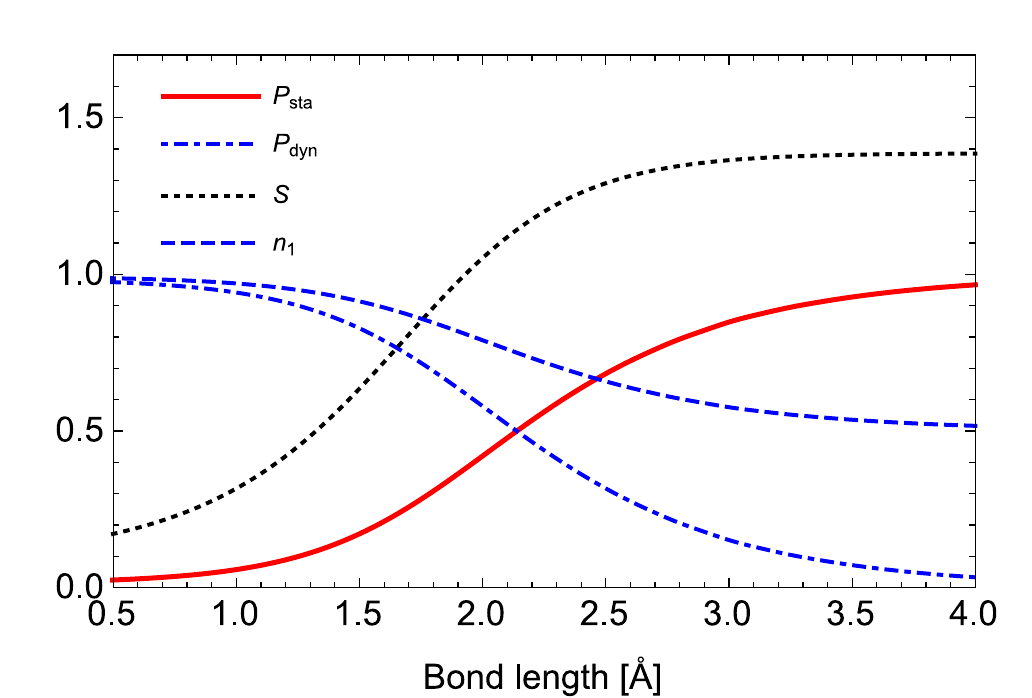}}
\subcaptionbox{Li$_2$.}
        {\includegraphics[width=8cm]{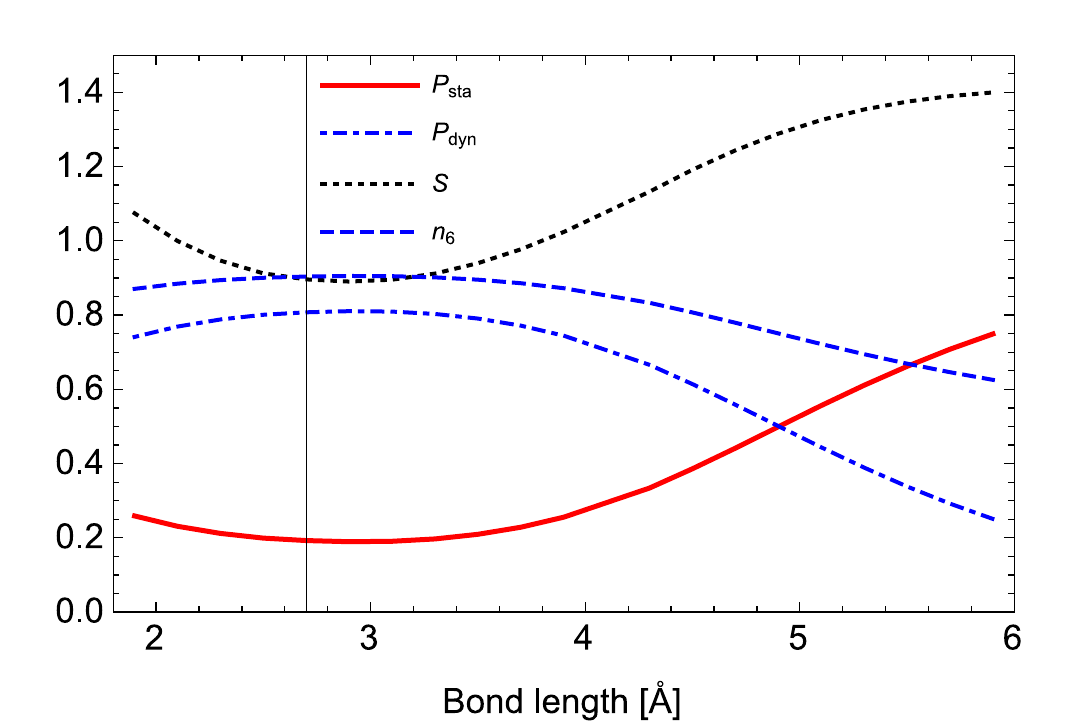}}
\caption{Correlation curves for H$_2$ and Li$_2$. 
 $P_{\rm sta}$ and $P_{\rm dyn}$ as
  well as the occupancy of the highest occupied natural 
  spin-orbital and the von Neumann entropy $S$ are 
  plotted as functions of the interatomic distance (in \AA).
  The equilibrium bonding length is 0.74  \AA~for H$_2$ and 2.67 
   \AA~for  Li$_2$.}
 \label{graf:H2}
 \end{figure}

The molecule  H$_2$, and in particular its dissociation limit, is the
 quintessential example of static correlation~\cite{CouFisch}. 
 It is well known that the restricted Hartree-Fock approach 
 describes very well the equilibrium chemical bond, but fails 
 dramatically as the molecule is stretched. Around the equilibrium
 separation, $P_{\rm sta}$ is close to
 zero and $P_{\rm dyn}$ reaches its maximum. There is 
 a change of regime around $1.5$ {\AA} because  the static 
 correlation begins to grow rapidly.
 Beyond this point, restricted Hartree-Fock theory is unable to 
 predict a bound system anymore. At the dissociation 
 limit, the correlation is due to static effects only, as expected. 
 Both measures allow us to observe the smooth increasing 
 of static effects when the 
 molecule is elongated.  A different situation can be observed 
 for the diatomic Li$_2$. For lengths smaller than the bond length
  the static correlation
 decreases as the distance increases.  
 The energy, the static correlation and the von-Neumann
 entropy reach their minimum around $2.9$ \AA, 
 very close to the equilibrium bonding length, while the 
 dynamic correlation as well as the occupancy of the highest 
 occupied natural spin-orbital acquired their maximum. Beyond that 
 value, the static correlation grows and 
 the dynamic correlation decreases slowly. This behaviour
 changes around $4$ \AA, where the static correlation speeds up.
 
 \begin{figure}[!t] 
\centering
      \subcaptionbox{H$_3$.}
        {\includegraphics[width=8cm]{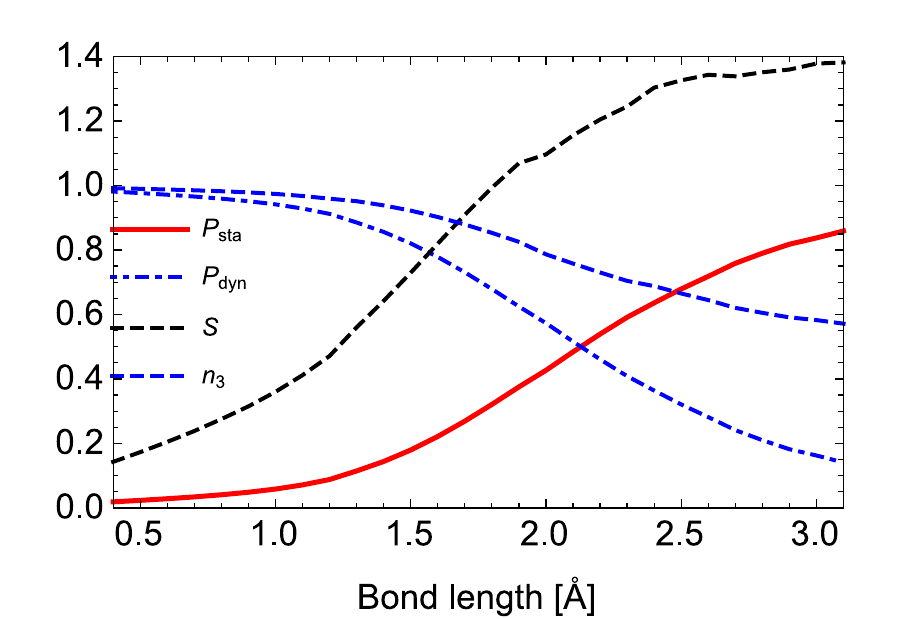}}
\subcaptionbox{Three-fermion three-site Hubbard model.}
        {\includegraphics[width=8cm]{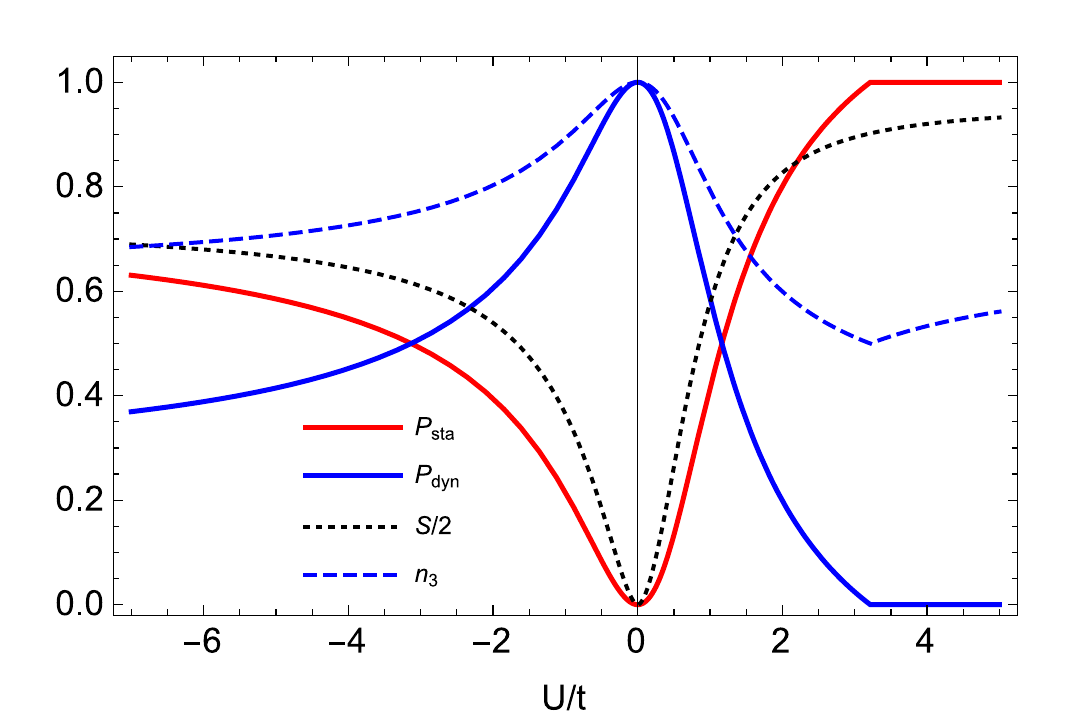}}
\caption{Correlation curves for H$_3$ and the 
three-fermion three-site Hubbard model. 
 $P_{\rm sta}$ and $P_{\rm dyn}$ as
  well as the occupancy of the highest occupied natural 
  spin-orbital and the von Neumann entropy $S$ are 
  plotted as functions of  the interatomic distance (in \AA)
  or the coupling strength.}
 \label{graf:H3}
 \end{figure}
 
 In Fig.~\ref{graf:H3} we plot the correlation measures
 for three-electron systems: the ground state of the
 equilateral H$_3$ and the three-site three-fermion Hubbard 
 model, which is very well known for it is analytically solvable \cite{newpaper, CS2015Hubbard}.  The Hamiltonian (in second quantization) of the one-dimensional $r$-site Hubbard model 
 reads:
\begin{align}
  \label{eq:Hamilt}
\hat H = -\frac{t}2 \sum_{i,\sigma}
(c^\dagger_{i\sigma} c_{(i+1)\sigma}
+ h.c. )
+ 2 U \sum_{i} \hat n_{i\uparrow} \hat n_{i \downarrow},
\end{align}
$i \in \{1,2,\dots,r\}$,
where $c^\dagger_{i\sigma}$ and $c_{i\sigma}$ are the fermionic
creation and annihilation operators for a particle on the site $i$
with spin $\sigma \in \{\uparrow, \downarrow \}$ and $\hat n_{i\sigma}
= c^\dagger_{i\sigma} c_{i\sigma}$.
The first term in Eq.~\eqref{eq:Hamilt} describes the hopping between
two neighboring sites while the second represents the on-site interaction.
Periodic boundary conditions for the case $r = 3$ are also assumed.
Achieved experimentally very recently with full control over the quantum state
\cite{PhysRevLett114}, this model may be considered as a simplified tight-binding description of the H$_r$ molecule.
 For the case of H$_3$ the correlation measures are 
 plotted as a function of the interatomic distance (in  \AA) and 
 for the Hubbard model as a function of the coupling $U/t$.  
 In both cases, as the molecule is elongated or the interaction 
 in the Hubbard model is enhanced, the energy gap 
 (the energy difference between the first-excited and the ground
 states) shortens and the electronic correlation increases, leading 
 to the appearance of static effects \cite{newpaper}.
 While H$_3$ exhibits a behaviour essentially similar to H$_2$, 
 the Hubbard model shows off two different 
 regimes of correlation. For positive values of the relative 
 coupling, the static correlation plays a prominent role. In
 particular, beyond $U/t = 3.2147$, the system lies in the 
 hyperplane $\mathcal{A}_2$ of the Borland-Dennis setting
\eqref{eq:const2} and the correlation is 
 completely static. For this strongly correlated regime, 
 the ground state can be written as a equiponderant
 superposition of three Slater determinants. For negative 
 values of the coupling there is always a fraction of the 
 correlation due to dynamic effects. In that limit the ground state
 is written as a superposition of two Slater determinants with different 
 amplitudes.
 
It is known that static and dynamic electronic correlations play
a prominent role in the orthogonally twisted ethylene \cite{Zen}. 
In fact, the energy gap between the ground and the first excited 
state shortens when the torsion angle around the C$=$C double 
bond is increased. While the ground state of the planar 
ethylene is very well described by a single Slater determinant,
at ninety degrees at least two Slater determinants 
are needed, resembling the dihydrogen in the dissociation limit. 
In Fig.~\ref{graf:ethyleneenergy} we plot the correlation energy 
and the energy gap of ethylene as a function of the
torsion angle, using CAS-SCF(12,12) method and a 
cc-pVDZ basis set. 
  In Fig.~\ref{graf:ethyl} we plot the correlation measures
  for ethylene along the torsional path. 
  For the planar geometry the correlation is almost completely
  dynamic and the situation remains in this way until the 
  torsional degree reaches 60$^{\rm o}$. From this angle on the 
  static correlation shows up. At  80$^{\rm o}$ the static and 
  dynamic correlation are equally important in the total 
  electron correlation of ethylene. When orthogonally twisted,
  the correlation of ethylene is 90\% due to static effects and 
  there is still an important part due to dynamic correlation. 
  Remarkably, the rise of static correlation around 60$^{\rm o}$
  coincides with the increase of correlation energy. From this 
  perspective, the gain of total correlation is mainly due to static 
  effects. 
  
 \begin{figure}[!t] 
 \centering
\includegraphics[width=8.5cm]{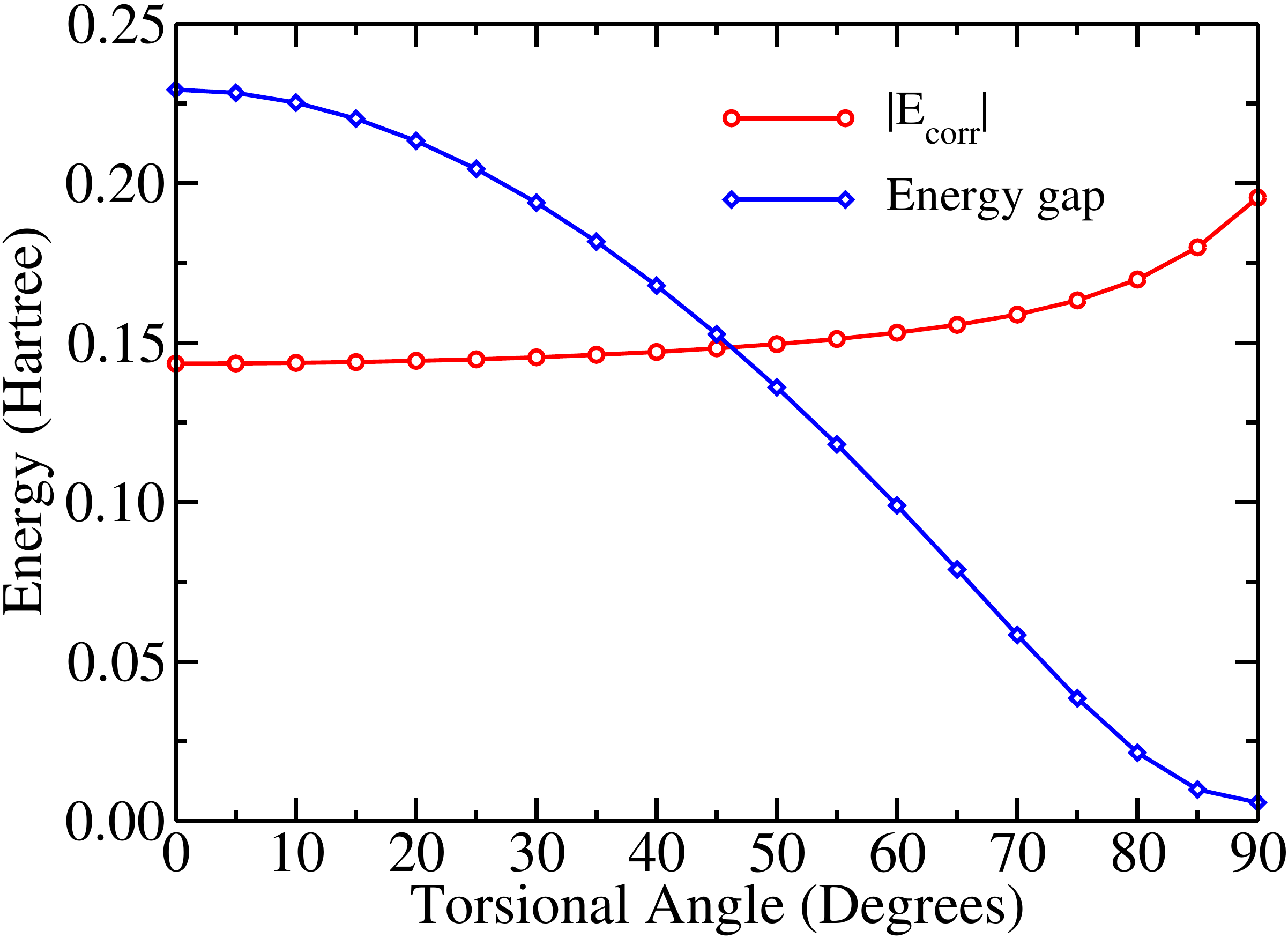}
\caption{Correlation energy and energy gap for the twisted 
ethylene C$_2$H$_4$ as a function of the torsion angle 
around the C=C double bond.}
 \label{graf:ethyleneenergy}
 \end{figure}

 \begin{figure}[!t] 
 \centering
\includegraphics[width=8.5cm]{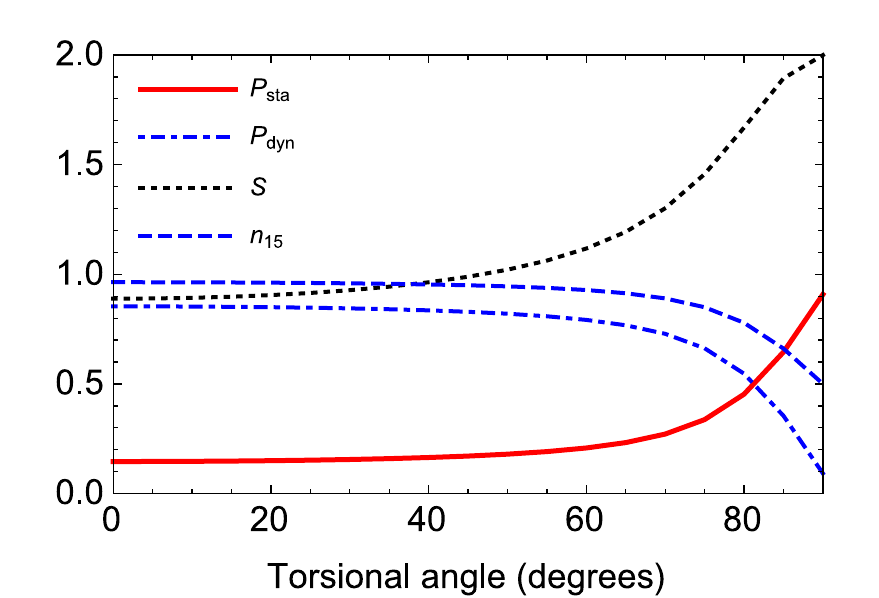}
\caption{Correlation curves for
ethylene. $P_{\rm sta}$ and $P_{\rm dyn}$ as
  well as the occupancy of the highest occupied natural 
  spin-orbital and the von Neumann entropy $S$ are 
  plotted as functions of  the torsion angle
  around the C=C double bond.}
 \label{graf:ethyl}
 \end{figure}
  
\section{Summary and conclusion}

Thanks to the generalization of the Pauli exclusion 
principle, it is possible to relate equiponderant 
superpositions of Slater determinants to certain sets of 
fermionic occupation numbers lying inside the Paulitope. 
In this paper, we have proposed a general criterion 
to distinguish the static and dynamic parts of the electronic 
correlation in fermionic systems, by tackling only one-particle 
information. By doing so, we provided 
two kinds of $l^1$-distances: (a)
to the Hartree-Fock point, which can be viewed as a measure 
of the dynamic part of the electronic correlation, and (b) to 
the static states, which can be viewed as a measure of 
the static part of the correlation.
We gave some examples of physical systems
and showed that these correlation measures correlate 
well with our intuition of static and dynamic correlation.

Though we focused our attention on two and three 
``active''-fermion systems, the results can in principle be 
generalized to larger settings. So far, the complete set of 
generalized Pauli constraints is only known for small systems
with three, four and five particles. There is, however, an algorithm 
which provide in principle the representability conditions
for larger settings \cite{Kly2}.

In this paper we have highlighted the paramount 
importance of the occupancy of the highest 
occupied natural spin-orbital in the understanding of 
the static correlation. The quantities we proposed in this 
paper can allow us to construct reliable ways to separate 
dynamic and static correlations and, more important, 
to better understand the qualitative nature of the correlation 
present in real physical and chemical electronic systems. 
They can also be a tool for analysing the failures of quantum 
many body theories (like density functional theory)
\cite{Chai}. Recent progress in fermionic mode entanglement 
can also shed more light in these directions \cite{doi:10.1021/ct400247p,Friis}.

\section*{acknowledgement}

We thank D. Gross, C. Schilling and M. Springborg for 
helpful discussions. We acknowledge financial support from 
the GSRT of the Hellenic Ministry of Education (ESPA), 
through ``Advanced Materials and Devices'' program 
(MIS:5002409) (N.N.L.) and the DFG through Projects 
No. SFB-762 and No. MA 6787/1-1 (M.A.L.M.).

\appendix 
\begin{center}
\textbf{\large Appendix}
\end{center}

\setcounter{equation}{0}
\setcounter{figure}{0}
\setcounter{table}{0}
\setcounter{section}{0}
\makeatletter

\begin{lem}\label{lem}
Let $\hat{H}$ be a Hamiltonian on the Hilbert space
$\H$ with a unique ground state $\ket{\phi_0}$ with 
energy $e_{\rm 0}$ and an energy gap 
$e_{\rm gap} = e_{\rm 1}-e_{\rm 0}$, where $e_1$ is
the energy of the first excited state. Then, for any 
$\ket{\psi} \in \H$ with energy $E_{\psi} = \bra{\psi}\hat{H}\ket{\psi}$ 
we have \cite{newpaper}:
$|\bra{\phi_{\rm 0}}\psi\rangle|^2 \geq (e_{\rm 1} - E_{\psi})/
e_{\rm gap}$.
\end{lem}

\begin{thm}\label{thm2}
Let $\hat{H}$ be a Hamiltonian on $\H$
with a unique ground-state. 
Let $S_E$ be the set of 
pure states with expected energy $E$:
$S_E = \{\psi \in \H |  \bra{\psi}\hat H\ket{\psi} = E \}$.
If $|E_1 - E_2| < \epsilon$, then $S_{E_2}$ has an element 
close to an element of $S_{E_1}$.
\end{thm}

\begin{proof}
Let us write the spectral decomposition of the Hamiltonian in 
the following way:
$\hat H = \sum_i e_i \ket{\phi_i}\bra{\phi_i}$,
with $e_0 < e_1 \leq e_2 \leq \cdots$. 
A wavefunction $\ket{\psi_1} \in S_{E_1}$
can be written in the eigenbasis of  $\hat H$ as 
$
\ket{\psi_1} = \sum_i a_i \ket{\phi_i} .
$ 
If $E_2 = E_1$ there is nothing to prove. 
Without loss of generality, 
let us assume that $E_2 < E_1$, take $0 \leq a_0 \leq 1$
and choose a state $\ket{\psi_2}$ in  $S_{E_2}$ as a 
superposition of $\ket{\psi_1}$ and the ground state
$\ket{\phi_0}$, namely:
$
\ket{\psi_2} = \alpha \ket{\psi_1} + \beta \ket{\phi_0}
$
with $\alpha$ and $\beta$ positive real numbers smaller than 1.
Normalization dictates that 
$
\alpha^2 + \beta^2 + 2\alpha \beta a_0 = 1
$
and the energy constraint reads 
$\alpha^2 E_1 + \beta^2 e_0 + 2\alpha \beta 
a_0 e_0 = E_2$. It is easy to see that both 
conditions translate into:
$\alpha = \sqrt{(E_2-e_0)/(E_1-e_0)}$
and
$\beta = [\alpha^2 a_0^2 + (1- \alpha^2)]^{1/2}- \alpha a_0$.
Using the fact that $0 \leq \alpha, \beta, a_0 \leq 1$
and Lemma \ref{lem} one obtains (for $E_1 \leq e_1$): 
\begin{align}
a_0 \beta  &= a_0\sqrt{\alpha^2 a_0^2 + (1- \alpha^2)}- \alpha a_0^2 
\geq a_0\sqrt{\alpha^2 a_0^2 + (1- \alpha^2)} - \alpha a_0
\nonumber
\\
&\geq \sqrt{\frac{e_1-E_1}{e_1-e_0}}
\Bigg(\sqrt{\frac{e_1-E_2}{e_1-e_0}}-
\sqrt{\frac{E_2-e_0}{E_1-e_0}}\Bigg)
\nonumber
\equiv f(e_0,e_1,E_1,E_2).
\end{align}
Now we can compare  the states $\ket{\psi_1}$ and $\ket{\psi_2}$:
\begin{align*}
||\psi_1 - \psi_2||^2 
\leq 2\sqrt{\frac{E_1-E_2}{E_1-e_0}} - 2\theta(e_1-E_1)
g(e_0,e_1,E_1,E_2),
\end{align*}
where the Heaviside function reads
$\theta(x) = 0$ if $x <0$ and $\theta(x) = 1$ if $x \geq 0$
and $g(\cdot) = \max(0,  f(\cdot))$.
\end{proof}


\bibliography{achemso-demo}

\begin{thebibliography}{72}%
\makeatletter
\providecommand \@ifxundefined [1]{%
 \@ifx{#1\undefined}
}%
\providecommand \@ifnum [1]{%
 \ifnum #1\expandafter \@firstoftwo
 \else \expandafter \@secondoftwo
 \fi
}%
\providecommand \@ifx [1]{%
 \ifx #1\expandafter \@firstoftwo
 \else \expandafter \@secondoftwo
 \fi
}%
\providecommand \natexlab [1]{#1}%
\providecommand \enquote  [1]{``#1''}%
\providecommand \bibnamefont  [1]{#1}%
\providecommand \bibfnamefont [1]{#1}%
\providecommand \citenamefont [1]{#1}%
\providecommand \href@noop [0]{\@secondoftwo}%
\providecommand \href [0]{\begingroup \@sanitize@url \@href}%
\providecommand \@href[1]{\@@startlink{#1}\@@href}%
\providecommand \@@href[1]{\endgroup#1\@@endlink}%
\providecommand \@sanitize@url [0]{\catcode `\\12\catcode `\$12\catcode
  `\&12\catcode `\#12\catcode `\^12\catcode `\_12\catcode `\%12\relax}%
\providecommand \@@startlink[1]{}%
\providecommand \@@endlink[0]{}%
\providecommand \url  [0]{\begingroup\@sanitize@url \@url }%
\providecommand \@url [1]{\endgroup\@href {#1}{\urlprefix }}%
\providecommand \urlprefix  [0]{URL }%
\providecommand \Eprint [0]{\href }%
\providecommand \doibase [0]{http://dx.doi.org/}%
\providecommand \selectlanguage [0]{\@gobble}%
\providecommand \bibinfo  [0]{\@secondoftwo}%
\providecommand \bibfield  [0]{\@secondoftwo}%
\providecommand \translation [1]{[#1]}%
\providecommand \BibitemOpen [0]{}%
\providecommand \bibitemStop [0]{}%
\providecommand \bibitemNoStop [0]{.\EOS\space}%
\providecommand \EOS [0]{\spacefactor3000\relax}%
\providecommand \BibitemShut  [1]{\csname bibitem#1\endcsname}%
\let\auto@bib@innerbib\@empty
%</preamble>
\bibitem [{\citenamefont {Tew}, \citenamefont {Klopper},\ and\ \citenamefont
  {Helgaker}(2007)}]{TewKlopperHelgaker}%
  \BibitemOpen
  \bibfield  {author} {\bibinfo {author} {\bibfnamefont {D.~P.}\ \bibnamefont
  {Tew}}, \bibinfo {author} {\bibfnamefont {W.}~\bibnamefont {Klopper}}, \ and\
  \bibinfo {author} {\bibfnamefont {T.}~\bibnamefont {Helgaker}},\ }\href
  {\doibase 10.1002/jcc.20581} {\bibfield  {journal} {\bibinfo  {journal} {J.
  Comput. Chem.}\ }\textbf {\bibinfo {volume} {28}},\ \bibinfo {pages} {1307}
  (\bibinfo {year} {2007})}\BibitemShut {NoStop}%
\bibitem [{\citenamefont {Wigner}(1934)}]{WignerCorr}%
  \BibitemOpen
  \bibfield  {author} {\bibinfo {author} {\bibfnamefont {E.}~\bibnamefont
  {Wigner}},\ }\href {\doibase 10.1103/PhysRev.46.1002} {\bibfield  {journal}
  {\bibinfo  {journal} {Phys. Rev.}\ }\textbf {\bibinfo {volume} {46}},\
  \bibinfo {pages} {1002} (\bibinfo {year} {1934})}\BibitemShut {NoStop}%
\bibitem [{\citenamefont {L\"owdin}(1955)}]{Lowdin}%
  \BibitemOpen
  \bibfield  {author} {\bibinfo {author} {\bibfnamefont {P.-O.}\ \bibnamefont
  {L\"owdin}},\ }\href {\doibase 10.1103/PhysRev.97.1509} {\bibfield  {journal}
  {\bibinfo  {journal} {Phys. Rev.}\ }\textbf {\bibinfo {volume} {97}},\
  \bibinfo {pages} {1509} (\bibinfo {year} {1955})}\BibitemShut {NoStop}%
\bibitem [{\citenamefont {Casula}\ and\ \citenamefont
  {Sorella}(2003)}]{CasulaI}%
  \BibitemOpen
  \bibfield  {author} {\bibinfo {author} {\bibfnamefont {M.}~\bibnamefont
  {Casula}}\ and\ \bibinfo {author} {\bibfnamefont {S.}~\bibnamefont
  {Sorella}},\ }\href {\doibase 10.1063/1.1604379} {\bibfield  {journal}
  {\bibinfo  {journal} {J. Chem. Phys.}\ }\textbf {\bibinfo {volume} {119}},\
  \bibinfo {pages} {6500} (\bibinfo {year} {2003})}\BibitemShut {NoStop}%
\bibitem [{\citenamefont {Casula}, \citenamefont {Attaccalite},\ and\
  \citenamefont {Sorella}(2004)}]{CasulaII}%
  \BibitemOpen
  \bibfield  {author} {\bibinfo {author} {\bibfnamefont {M.}~\bibnamefont
  {Casula}}, \bibinfo {author} {\bibfnamefont {C.}~\bibnamefont {Attaccalite}},
  \ and\ \bibinfo {author} {\bibfnamefont {S.}~\bibnamefont {Sorella}},\ }\href
  {\doibase 10.1063/1.1794632} {\bibfield  {journal} {\bibinfo  {journal} {J.
  Chem. Phys.}\ }\textbf {\bibinfo {volume} {121}},\ \bibinfo {pages} {7110}
  (\bibinfo {year} {2004})}\BibitemShut {NoStop}%
\bibitem [{\citenamefont {Zen}\ \emph {et~al.}(2014)\citenamefont {Zen},
  \citenamefont {Coccia}, \citenamefont {Luo}, \citenamefont {Sorella},\ and\
  \citenamefont {Guidoni}}]{Zen}%
  \BibitemOpen
  \bibfield  {author} {\bibinfo {author} {\bibfnamefont {A.}~\bibnamefont
  {Zen}}, \bibinfo {author} {\bibfnamefont {E.}~\bibnamefont {Coccia}},
  \bibinfo {author} {\bibfnamefont {Y.}~\bibnamefont {Luo}}, \bibinfo {author}
  {\bibfnamefont {S.}~\bibnamefont {Sorella}}, \ and\ \bibinfo {author}
  {\bibfnamefont {L.}~\bibnamefont {Guidoni}},\ }\href {\doibase
  10.1021/ct401008s} {\bibfield  {journal} {\bibinfo  {journal} {J. Chem.
  Theory Comput.}\ }\textbf {\bibinfo {volume} {10}},\ \bibinfo {pages} {1048}
  (\bibinfo {year} {2014})}\BibitemShut {NoStop}%
\bibitem [{\citenamefont {Neuscamman}(2013)}]{Neuscamman}%
  \BibitemOpen
  \bibfield  {author} {\bibinfo {author} {\bibfnamefont {E.}~\bibnamefont
  {Neuscamman}},\ }\href {\doibase 10.1063/1.4829835} {\bibfield  {journal}
  {\bibinfo  {journal} {J. Chem. Phys.}\ }\textbf {\bibinfo {volume} {139}},\
  \bibinfo {pages} {194105} (\bibinfo {year} {2013})}\BibitemShut {NoStop}%
\bibitem [{\citenamefont {Schliemann}\ \emph {et~al.}(2001)\citenamefont
  {Schliemann}, \citenamefont {Cirac}, \citenamefont {Ku\ifmmode~\acute{s}\else
  \'{s}\fi{}}, \citenamefont {Lewenstein},\ and\ \citenamefont {Loss}}]{Cirac}%
  \BibitemOpen
  \bibfield  {author} {\bibinfo {author} {\bibfnamefont {J.}~\bibnamefont
  {Schliemann}}, \bibinfo {author} {\bibfnamefont {J.~I.}\ \bibnamefont
  {Cirac}}, \bibinfo {author} {\bibfnamefont {M.}~\bibnamefont
  {Ku\ifmmode~\acute{s}\else \'{s}\fi{}}}, \bibinfo {author} {\bibfnamefont
  {M.}~\bibnamefont {Lewenstein}}, \ and\ \bibinfo {author} {\bibfnamefont
  {D.}~\bibnamefont {Loss}},\ }\href {\doibase 10.1103/PhysRevA.64.022303}
  {\bibfield  {journal} {\bibinfo  {journal} {Phys. Rev. A}\ }\textbf {\bibinfo
  {volume} {64}},\ \bibinfo {pages} {022303} (\bibinfo {year}
  {2001})}\BibitemShut {NoStop}%
\bibitem [{\citenamefont {Plastino}, \citenamefont {Manzano},\ and\
  \citenamefont {Dehesa}(2009)}]{Plastino}%
  \BibitemOpen
  \bibfield  {author} {\bibinfo {author} {\bibfnamefont {A.~R.}\ \bibnamefont
  {Plastino}}, \bibinfo {author} {\bibfnamefont {D.}~\bibnamefont {Manzano}}, \
  and\ \bibinfo {author} {\bibfnamefont {J.~S.}\ \bibnamefont {Dehesa}},\
  }\href {http://stacks.iop.org/0295-5075/86/i=2/a=20005} {\bibfield  {journal}
  {\bibinfo  {journal} {EPL}\ }\textbf {\bibinfo {volume} {86}},\ \bibinfo
  {pages} {20005} (\bibinfo {year} {2009})}\BibitemShut {NoStop}%
\bibitem [{\citenamefont {S\'arosi}\ and\ \citenamefont
  {L\'evay}(2014)}]{Magyares}%
  \BibitemOpen
  \bibfield  {author} {\bibinfo {author} {\bibfnamefont {G.}~\bibnamefont
  {S\'arosi}}\ and\ \bibinfo {author} {\bibfnamefont {P.}~\bibnamefont
  {L\'evay}},\ }\href {\doibase 10.1103/PhysRevA.89.042310} {\bibfield
  {journal} {\bibinfo  {journal} {Phys. Rev. A}\ }\textbf {\bibinfo {volume}
  {89}},\ \bibinfo {pages} {042310} (\bibinfo {year} {2014})}\BibitemShut
  {NoStop}%
\bibitem [{\citenamefont {Juh\'asz}\ and\ \citenamefont
  {Mazziotti}(2006)}]{Mazziotticorr}%
  \BibitemOpen
  \bibfield  {author} {\bibinfo {author} {\bibfnamefont {T.}~\bibnamefont
  {Juh\'asz}}\ and\ \bibinfo {author} {\bibfnamefont {D.~A.}\ \bibnamefont
  {Mazziotti}},\ }\href
  {http://scitation.aip.org/content/aip/journal/jcp/125/17/10.1063/1.2378768}
  {\bibfield  {journal} {\bibinfo  {journal} {J. Chem. Phys.}\ }\textbf
  {\bibinfo {volume} {125}},\ \bibinfo {pages} {174105} (\bibinfo {year}
  {2006})}\BibitemShut {NoStop}%
\bibitem [{\citenamefont {Gottlieb}\ and\ \citenamefont
  {Mauser}(2005)}]{PhysRevLett.95.123003}%
  \BibitemOpen
  \bibfield  {author} {\bibinfo {author} {\bibfnamefont {A.~D.}\ \bibnamefont
  {Gottlieb}}\ and\ \bibinfo {author} {\bibfnamefont {N.~J.}\ \bibnamefont
  {Mauser}},\ }\href {\doibase 10.1103/PhysRevLett.95.123003} {\bibfield
  {journal} {\bibinfo  {journal} {Phys. Rev. Lett.}\ }\textbf {\bibinfo
  {volume} {95}},\ \bibinfo {pages} {123003} (\bibinfo {year}
  {2005})}\BibitemShut {NoStop}%
\bibitem [{\citenamefont {Benavides-Riveros}\ \emph {et~al.}(2017)\citenamefont
  {Benavides-Riveros}, \citenamefont {Lathiotakis}, \citenamefont {Schilling},\
  and\ \citenamefont {Marques}}]{newpaper}%
  \BibitemOpen
  \bibfield  {author} {\bibinfo {author} {\bibfnamefont {C.~L.}\ \bibnamefont
  {Benavides-Riveros}}, \bibinfo {author} {\bibfnamefont {N.~N.}\ \bibnamefont
  {Lathiotakis}}, \bibinfo {author} {\bibfnamefont {C.}~\bibnamefont
  {Schilling}}, \ and\ \bibinfo {author} {\bibfnamefont {M.~A.~L.}\
  \bibnamefont {Marques}},\ }\href {\doibase 10.1103/PhysRevA.95.032507}
  {\bibfield  {journal} {\bibinfo  {journal} {Phys. Rev. A}\ }\textbf {\bibinfo
  {volume} {95}},\ \bibinfo {pages} {032507} (\bibinfo {year}
  {2017})}\BibitemShut {NoStop}%
\bibitem [{\citenamefont {Becke}(2013)}]{Becke}%
  \BibitemOpen
  \bibfield  {author} {\bibinfo {author} {\bibfnamefont {A.~D.}\ \bibnamefont
  {Becke}},\ }\href
  {http://scitation.aip.org/content/aip/journal/jcp/138/7/10.1063/1.4790598}
  {\bibfield  {journal} {\bibinfo  {journal} {J. Chem. Phys.}\ }\textbf
  {\bibinfo {volume} {138}},\ \bibinfo {pages} {074109} (\bibinfo {year}
  {2013})}\BibitemShut {NoStop}%
\bibitem [{\citenamefont {Ziesche}\ \emph {et~al.}(1997)\citenamefont
  {Ziesche}, \citenamefont {Gunnarsson}, \citenamefont {John},\ and\
  \citenamefont {Beck}}]{Ziesche}%
  \BibitemOpen
  \bibfield  {author} {\bibinfo {author} {\bibfnamefont {P.}~\bibnamefont
  {Ziesche}}, \bibinfo {author} {\bibfnamefont {O.}~\bibnamefont {Gunnarsson}},
  \bibinfo {author} {\bibfnamefont {W.}~\bibnamefont {John}}, \ and\ \bibinfo
  {author} {\bibfnamefont {H.}~\bibnamefont {Beck}},\ }\href {\doibase
  10.1103/PhysRevB.55.10270} {\bibfield  {journal} {\bibinfo  {journal} {Phys.
  Rev. B}\ }\textbf {\bibinfo {volume} {55}},\ \bibinfo {pages} {10270}
  (\bibinfo {year} {1997})}\BibitemShut {NoStop}%
\bibitem [{\citenamefont {Cohen}, \citenamefont {Mori-S{\'a}nchez},\ and\
  \citenamefont {Yang}(2008{\natexlab{a}})}]{Cohen792}%
  \BibitemOpen
  \bibfield  {author} {\bibinfo {author} {\bibfnamefont {A.~J.}\ \bibnamefont
  {Cohen}}, \bibinfo {author} {\bibfnamefont {P.}~\bibnamefont
  {Mori-S{\'a}nchez}}, \ and\ \bibinfo {author} {\bibfnamefont
  {W.}~\bibnamefont {Yang}},\ }\href {\doibase 10.1126/science.1158722}
  {\bibfield  {journal} {\bibinfo  {journal} {Science}\ }\textbf {\bibinfo
  {volume} {321}},\ \bibinfo {pages} {792} (\bibinfo {year}
  {2008}{\natexlab{a}})}\BibitemShut {NoStop}%
\bibitem [{\citenamefont {Cohen}, \citenamefont {Mori-S{\'a}nchez},\ and\
  \citenamefont {Yang}(2008{\natexlab{b}})}]{Sanchez}%
  \BibitemOpen
  \bibfield  {author} {\bibinfo {author} {\bibfnamefont {A.~J.}\ \bibnamefont
  {Cohen}}, \bibinfo {author} {\bibfnamefont {P.}~\bibnamefont
  {Mori-S{\'a}nchez}}, \ and\ \bibinfo {author} {\bibfnamefont
  {W.}~\bibnamefont {Yang}},\ }\href
  {http://scitation.aip.org/content/aip/journal/jcp/129/12/10.1063/1.2987202}
  {\bibfield  {journal} {\bibinfo  {journal} {J. Chem. Phys.}\ }\textbf
  {\bibinfo {volume} {129}},\ \bibinfo {pages} {121104} (\bibinfo {year}
  {2008}{\natexlab{b}})}\BibitemShut {NoStop}%
\bibitem [{\citenamefont {Hollett}\ and\ \citenamefont {Gill}(2011)}]{Hollett}%
  \BibitemOpen
  \bibfield  {author} {\bibinfo {author} {\bibfnamefont {J.~W.}\ \bibnamefont
  {Hollett}}\ and\ \bibinfo {author} {\bibfnamefont {P.~M.~W.}\ \bibnamefont
  {Gill}},\ }\href {\doibase 10.1063/1.3570574} {\bibfield  {journal} {\bibinfo
   {journal} {J. Chem. Phys.}\ }\textbf {\bibinfo {volume} {134}},\ \bibinfo
  {pages} {114111} (\bibinfo {year} {2011})}\BibitemShut {NoStop}%
\bibitem [{\citenamefont {Hollett}, \citenamefont {Hosseini},\ and\
  \citenamefont {Menzies}(2016)}]{cumulant}%
  \BibitemOpen
  \bibfield  {author} {\bibinfo {author} {\bibfnamefont {J.~W.}\ \bibnamefont
  {Hollett}}, \bibinfo {author} {\bibfnamefont {H.}~\bibnamefont {Hosseini}}, \
  and\ \bibinfo {author} {\bibfnamefont {C.}~\bibnamefont {Menzies}},\ }\href
  {http://scitation.aip.org/content/aip/journal/jcp/145/8/10.1063/1.4961243}
  {\bibfield  {journal} {\bibinfo  {journal} {J. Chem. Phys.}\ }\textbf
  {\bibinfo {volume} {145}},\ \bibinfo {pages} {084106} (\bibinfo {year}
  {2016})}\BibitemShut {NoStop}%
\bibitem [{\citenamefont {Walter}\ \emph {et~al.}(2013)\citenamefont {Walter},
  \citenamefont {Doran}, \citenamefont {Gross},\ and\ \citenamefont
  {Christandl}}]{Walter1205}%
  \BibitemOpen
  \bibfield  {author} {\bibinfo {author} {\bibfnamefont {M.}~\bibnamefont
  {Walter}}, \bibinfo {author} {\bibfnamefont {B.}~\bibnamefont {Doran}},
  \bibinfo {author} {\bibfnamefont {D.}~\bibnamefont {Gross}}, \ and\ \bibinfo
  {author} {\bibfnamefont {M.}~\bibnamefont {Christandl}},\ }\href
  {http://science.sciencemag.org/content/340/6137/1205} {\bibfield  {journal}
  {\bibinfo  {journal} {Science}\ }\textbf {\bibinfo {volume} {340}},\ \bibinfo
  {pages} {1205} (\bibinfo {year} {2013})}\BibitemShut {NoStop}%
\bibitem [{\citenamefont {Sawicki}, \citenamefont {Walter},\ and\ \citenamefont
  {Ku\ifmmode~\acute{s}\else \'{s}\fi{}}(2013)}]{Sawicki}%
  \BibitemOpen
  \bibfield  {author} {\bibinfo {author} {\bibfnamefont {A.}~\bibnamefont
  {Sawicki}}, \bibinfo {author} {\bibfnamefont {M.}~\bibnamefont {Walter}}, \
  and\ \bibinfo {author} {\bibfnamefont {M.}~\bibnamefont
  {Ku\ifmmode~\acute{s}\else \'{s}\fi{}}},\ }\href
  {http://stacks.iop.org/1751-8121/46/i=5/a=055304} {\bibfield  {journal}
  {\bibinfo  {journal} {J. Phys. A}\ }\textbf {\bibinfo {volume} {46}},\
  \bibinfo {pages} {055304} (\bibinfo {year} {2013})}\BibitemShut {NoStop}%
\bibitem [{\citenamefont {Ramos-Cordoba}, \citenamefont {Salvador},\ and\
  \citenamefont {Matito}(2016)}]{Matito}%
  \BibitemOpen
  \bibfield  {author} {\bibinfo {author} {\bibfnamefont {E.}~\bibnamefont
  {Ramos-Cordoba}}, \bibinfo {author} {\bibfnamefont {P.}~\bibnamefont
  {Salvador}}, \ and\ \bibinfo {author} {\bibfnamefont {E.}~\bibnamefont
  {Matito}},\ }\href {\doibase 10.1039/C6CP03072F} {\bibfield  {journal}
  {\bibinfo  {journal} {Phys. Chem. Chem. Phys.}\ }\textbf {\bibinfo {volume}
  {18}},\ \bibinfo {pages} {24015} (\bibinfo {year} {2016})}\BibitemShut
  {NoStop}%
\bibitem [{\citenamefont {Klyachko}(2006)}]{Kly2}%
  \BibitemOpen
  \bibfield  {author} {\bibinfo {author} {\bibfnamefont {A.}~\bibnamefont
  {Klyachko}},\ }\href {http://stacks.iop.org/1742-6596/36/i=1/a=014}
  {\bibfield  {journal} {\bibinfo  {journal} {J. Phys.}\ }\textbf {\bibinfo
  {volume} {36}},\ \bibinfo {pages} {72} (\bibinfo {year} {2006})}\BibitemShut
  {NoStop}%
\bibitem [{\citenamefont {Theophilou}\ \emph {et~al.}(2015)\citenamefont
  {Theophilou}, \citenamefont {Lathiotakis}, \citenamefont {Marques},\ and\
  \citenamefont {Helbig}}]{RDMFT}%
  \BibitemOpen
  \bibfield  {author} {\bibinfo {author} {\bibfnamefont {I.}~\bibnamefont
  {Theophilou}}, \bibinfo {author} {\bibfnamefont {N.}~\bibnamefont
  {Lathiotakis}}, \bibinfo {author} {\bibfnamefont {M.}~\bibnamefont
  {Marques}}, \ and\ \bibinfo {author} {\bibfnamefont {N.}~\bibnamefont
  {Helbig}},\ }\href
  {http://scitation.aip.org/content/aip/journal/jcp/142/15/10.1063/1.4918346}
  {\bibfield  {journal} {\bibinfo  {journal} {J. Chem. Phys.}\ }\textbf
  {\bibinfo {volume} {142}},\ \bibinfo {pages} {154108} (\bibinfo {year}
  {2015})}\BibitemShut {NoStop}%
\bibitem [{\citenamefont {Mazziotti}(2016)}]{recentMazziotti}%
  \BibitemOpen
  \bibfield  {author} {\bibinfo {author} {\bibfnamefont {D.~A.}\ \bibnamefont
  {Mazziotti}},\ }\href {\doibase 10.1103/PhysRevA.94.032516} {\bibfield
  {journal} {\bibinfo  {journal} {Phys. Rev. A}\ }\textbf {\bibinfo {volume}
  {94}},\ \bibinfo {pages} {032516} (\bibinfo {year} {2016})}\BibitemShut
  {NoStop}%
\bibitem [{\citenamefont {DePrince}(2016)}]{DePrince}%
  \BibitemOpen
  \bibfield  {author} {\bibinfo {author} {\bibfnamefont {A.~E.}\ \bibnamefont
  {DePrince}},\ }\href {\doibase 10.1063/1.4965888} {\bibfield  {journal}
  {\bibinfo  {journal} {J. Chem. Phys.}\ }\textbf {\bibinfo {volume} {145}},\
  \bibinfo {pages} {164109} (\bibinfo {year} {2016})}\BibitemShut {NoStop}%
\bibitem [{\citenamefont {M\"uller}(1999)}]{Mullern}%
  \BibitemOpen
  \bibfield  {author} {\bibinfo {author} {\bibfnamefont {C.~W.}\ \bibnamefont
  {M\"uller}},\ }\href {http://stacks.iop.org/0305-4470/32/i=22/a=314}
  {\bibfield  {journal} {\bibinfo  {journal} {J. Phys. A}\ }\textbf {\bibinfo
  {volume} {32}},\ \bibinfo {pages} {4139} (\bibinfo {year}
  {1999})}\BibitemShut {NoStop}%
\bibitem [{\citenamefont {Coleman}(1963)}]{Col2}%
  \BibitemOpen
  \bibfield  {author} {\bibinfo {author} {\bibfnamefont {A.~J.}\ \bibnamefont
  {Coleman}},\ }\href {\doibase 10.1103/RevModPhys.35.668} {\bibfield
  {journal} {\bibinfo  {journal} {Rev. Mod. Phys.}\ }\textbf {\bibinfo {volume}
  {35}},\ \bibinfo {pages} {668} (\bibinfo {year} {1963})}\BibitemShut
  {NoStop}%
\bibitem [{\citenamefont {Altunbulak}\ and\ \citenamefont
  {Klyachko}(2008)}]{Kly3}%
  \BibitemOpen
  \bibfield  {author} {\bibinfo {author} {\bibfnamefont {M.}~\bibnamefont
  {Altunbulak}}\ and\ \bibinfo {author} {\bibfnamefont {A.}~\bibnamefont
  {Klyachko}},\ }\href {\doibase 10.1007/s00220-008-0552-z} {\bibfield
  {journal} {\bibinfo  {journal} {Commun. Math. Phys.}\ }\textbf {\bibinfo
  {volume} {282}},\ \bibinfo {pages} {287} (\bibinfo {year}
  {2008})}\BibitemShut {NoStop}%
\bibitem [{\citenamefont {Schilling}, \citenamefont {Gross},\ and\
  \citenamefont {Christandl}(2013)}]{CS2013}%
  \BibitemOpen
  \bibfield  {author} {\bibinfo {author} {\bibfnamefont {C.}~\bibnamefont
  {Schilling}}, \bibinfo {author} {\bibfnamefont {D.}~\bibnamefont {Gross}}, \
  and\ \bibinfo {author} {\bibfnamefont {M.}~\bibnamefont {Christandl}},\
  }\href {\doibase 10.1103/PhysRevLett.110.040404} {\bibfield  {journal}
  {\bibinfo  {journal} {Phys. Rev. Lett.}\ }\textbf {\bibinfo {volume} {110}},\
  \bibinfo {pages} {040404} (\bibinfo {year} {2013})}\BibitemShut {NoStop}%
\bibitem [{\citenamefont {Schilling}(2014{\natexlab{a}})}]{CSQMath12}%
  \BibitemOpen
  \bibfield  {author} {\bibinfo {author} {\bibfnamefont {C.}~\bibnamefont
  {Schilling}},\ }in\ \href {\doibase 10.1142/9789814618144_0010} {\emph
  {\bibinfo {booktitle} {Mathematical Results in Quantum Mechanics}}},\
  \bibinfo {editor} {edited by\ \bibinfo {editor} {\bibfnamefont
  {P.}~\bibnamefont {Exner}}, \bibinfo {editor} {\bibfnamefont
  {W.}~\bibnamefont {K\"onig}}, \ and\ \bibinfo {editor} {\bibfnamefont
  {H.}~\bibnamefont {Neidhardt}}}\ (\bibinfo  {publisher} {World Scientific},\
  \bibinfo {year} {2014})\ Chap.~\bibinfo {chapter} {10}, pp.\ \bibinfo {pages}
  {165--176}\BibitemShut {NoStop}%
\bibitem [{Note1()}]{Note1}%
  \BibitemOpen
  \bibinfo {note} {Norbert Mauser coined the term ``Paulitope'' during the
  Workshop \protect \textit {Generalized Pauli Constraints and Fermion
  Correlation}, celebrated at the Wolfgang Pauli Institute in Vienna in August
  2016.}\BibitemShut {Stop}%
\bibitem [{\citenamefont {Schilling}(2015{\natexlab{a}})}]{CSQuasipinning}%
  \BibitemOpen
  \bibfield  {author} {\bibinfo {author} {\bibfnamefont {C.}~\bibnamefont
  {Schilling}},\ }\href {\doibase 10.1103/PhysRevA.91.022105} {\bibfield
  {journal} {\bibinfo  {journal} {Phys. Rev. A}\ }\textbf {\bibinfo {volume}
  {91}},\ \bibinfo {pages} {022105} (\bibinfo {year}
  {2015}{\natexlab{a}})}\BibitemShut {NoStop}%
\bibitem [{\citenamefont {Benavides-Riveros}, \citenamefont {Gracia-Bondia},\
  and\ \citenamefont {Spring\-borg}(2013)}]{BenavLiQuasi}%
  \BibitemOpen
  \bibfield  {author} {\bibinfo {author} {\bibfnamefont {C.~L.}\ \bibnamefont
  {Benavides-Riveros}}, \bibinfo {author} {\bibfnamefont {J.~M.}\ \bibnamefont
  {Gracia-Bondia}}, \ and\ \bibinfo {author} {\bibfnamefont {M.}~\bibnamefont
  {Spring\-borg}},\ }\href {\doibase 10.1103/PhysRevA.88.022508} {\bibfield
  {journal} {\bibinfo  {journal} {Phys. Rev. A}\ }\textbf {\bibinfo {volume}
  {88}},\ \bibinfo {pages} {022508} (\bibinfo {year} {2013})}\BibitemShut
  {NoStop}%
\bibitem [{\citenamefont {Chakraborty}\ and\ \citenamefont
  {Mazziotti}(2014)}]{Mazz14}%
  \BibitemOpen
  \bibfield  {author} {\bibinfo {author} {\bibfnamefont {R.}~\bibnamefont
  {Chakraborty}}\ and\ \bibinfo {author} {\bibfnamefont {D.}~\bibnamefont
  {Mazziotti}},\ }\href {\doibase 10.1103/PhysRevA.89.042505} {\bibfield
  {journal} {\bibinfo  {journal} {Phys. Rev. A}\ }\textbf {\bibinfo {volume}
  {89}},\ \bibinfo {pages} {042505} (\bibinfo {year} {2014})}\BibitemShut
  {NoStop}%
\bibitem [{\citenamefont {Benavides-Riveros}, \citenamefont
  {Gracia-Bond{\'\i}a},\ and\ \citenamefont {Springborg}(2014)}]{Benavdoubly}%
  \BibitemOpen
  \bibfield  {author} {\bibinfo {author} {\bibfnamefont {C.~L.}\ \bibnamefont
  {Benavides-Riveros}}, \bibinfo {author} {\bibfnamefont {J.~M.}\ \bibnamefont
  {Gracia-Bond{\'\i}a}}, \ and\ \bibinfo {author} {\bibfnamefont
  {M.}~\bibnamefont {Springborg}},\ }\href {https://arxiv.org/abs/1409.6435}
  {\bibfield  {journal} {\bibinfo  {journal} {arXiv:1409.6435}\ } (\bibinfo
  {year} {2014})}\BibitemShut {NoStop}%
\bibitem [{\citenamefont {Benavides-Riveros}\ and\ \citenamefont
  {Springborg}(2015)}]{BenavQuasi2}%
  \BibitemOpen
  \bibfield  {author} {\bibinfo {author} {\bibfnamefont {C.~L.}\ \bibnamefont
  {Benavides-Riveros}}\ and\ \bibinfo {author} {\bibfnamefont {M.}~\bibnamefont
  {Springborg}},\ }\href {\doibase 10.1103/PhysRevA.92.012512} {\bibfield
  {journal} {\bibinfo  {journal} {Phys. Rev. A}\ }\textbf {\bibinfo {volume}
  {92}},\ \bibinfo {pages} {012512} (\bibinfo {year} {2015})}\BibitemShut
  {NoStop}%
\bibitem [{\citenamefont {Chakraborty}\ and\ \citenamefont
  {Mazziotti}(2015{\natexlab{a}})}]{chakraborty2015structure}%
  \BibitemOpen
  \bibfield  {author} {\bibinfo {author} {\bibfnamefont {R.}~\bibnamefont
  {Chakraborty}}\ and\ \bibinfo {author} {\bibfnamefont {D.~A.}\ \bibnamefont
  {Mazziotti}},\ }\href {\doibase 10.1002/qua.24934} {\bibfield  {journal}
  {\bibinfo  {journal} {Int. J. Quantum Chem.}\ }\textbf {\bibinfo {volume}
  {115}},\ \bibinfo {pages} {1305} (\bibinfo {year}
  {2015}{\natexlab{a}})}\BibitemShut {NoStop}%
\bibitem [{\citenamefont {Benavides-Riveros}\ and\ \citenamefont
  {Schilling}(2016)}]{CSHFZPC}%
  \BibitemOpen
  \bibfield  {author} {\bibinfo {author} {\bibfnamefont {C.~L.}\ \bibnamefont
  {Benavides-Riveros}}\ and\ \bibinfo {author} {\bibfnamefont {C.}~\bibnamefont
  {Schilling}},\ }\href
  {http://www.degruyter.com/view/j/zpch.2016.230.issue-5-7/zpch-2015-0732/zpch-2015-0732.xml?format=INT}
  {\bibfield  {journal} {\bibinfo  {journal} {Z. Phys. Chem.}\ }\textbf
  {\bibinfo {volume} {230}},\ \bibinfo {pages} {703} (\bibinfo {year}
  {2016})}\BibitemShut {NoStop}%
\bibitem [{\citenamefont {Tennie}, \citenamefont {Vedral},\ and\ \citenamefont
  {Schilling}(2016)}]{CS2016b}%
  \BibitemOpen
  \bibfield  {author} {\bibinfo {author} {\bibfnamefont {F.}~\bibnamefont
  {Tennie}}, \bibinfo {author} {\bibfnamefont {V.}~\bibnamefont {Vedral}}, \
  and\ \bibinfo {author} {\bibfnamefont {C.}~\bibnamefont {Schilling}},\ }\href
  {\doibase 10.1103/PhysRevA.94.012120} {\bibfield  {journal} {\bibinfo
  {journal} {Phys. Rev. A}\ }\textbf {\bibinfo {volume} {94}},\ \bibinfo
  {pages} {012120} (\bibinfo {year} {2016})}\BibitemShut {NoStop}%
\bibitem [{\citenamefont {Tennie}\ \emph {et~al.}(2016)\citenamefont {Tennie},
  \citenamefont {Ebler}, \citenamefont {Vedral},\ and\ \citenamefont
  {Schilling}}]{TVS16}%
  \BibitemOpen
  \bibfield  {author} {\bibinfo {author} {\bibfnamefont {F.}~\bibnamefont
  {Tennie}}, \bibinfo {author} {\bibfnamefont {D.}~\bibnamefont {Ebler}},
  \bibinfo {author} {\bibfnamefont {V.}~\bibnamefont {Vedral}}, \ and\ \bibinfo
  {author} {\bibfnamefont {C.}~\bibnamefont {Schilling}},\ }\href {\doibase
  10.1103/PhysRevA.93.042126} {\bibfield  {journal} {\bibinfo  {journal} {Phys.
  Rev. A}\ }\textbf {\bibinfo {volume} {93}},\ \bibinfo {pages} {042126}
  (\bibinfo {year} {2016})}\BibitemShut {NoStop}%
\bibitem [{\citenamefont {Tennie}, \citenamefont {Vedral},\ and\ \citenamefont
  {Schilling}(2017)}]{TVS17}%
  \BibitemOpen
  \bibfield  {author} {\bibinfo {author} {\bibfnamefont {F.}~\bibnamefont
  {Tennie}}, \bibinfo {author} {\bibfnamefont {V.}~\bibnamefont {Vedral}}, \
  and\ \bibinfo {author} {\bibfnamefont {C.}~\bibnamefont {Schilling}},\ }\href
  {\doibase 10.1103/PhysRevA.95.022336} {\bibfield  {journal} {\bibinfo
  {journal} {Phys. Rev. A}\ }\textbf {\bibinfo {volume} {95}},\ \bibinfo
  {pages} {022336} (\bibinfo {year} {2017})}\BibitemShut {NoStop}%
\bibitem [{\citenamefont {Wang}, \citenamefont {Wang},\ and\ \citenamefont
  {Lischka}(2017)}]{QUA:QUA25376}%
  \BibitemOpen
  \bibfield  {author} {\bibinfo {author} {\bibfnamefont {Y.}~\bibnamefont
  {Wang}}, \bibinfo {author} {\bibfnamefont {J.}~\bibnamefont {Wang}}, \ and\
  \bibinfo {author} {\bibfnamefont {H.}~\bibnamefont {Lischka}},\ }\href
  {\doibase 10.1002/qua.25376} {\bibfield  {journal} {\bibinfo  {journal} {Int.
  J. Quantum Chem.}\ ,\ \bibinfo {pages} {e25376}} (\bibinfo {year} {2017})},\
  \bibinfo {note} {e25376}\BibitemShut {NoStop}%
\bibitem [{\citenamefont {Chakraborty}\ and\ \citenamefont
  {Mazziotti}(2015{\natexlab{b}})}]{MazzOpen}%
  \BibitemOpen
  \bibfield  {author} {\bibinfo {author} {\bibfnamefont {R.}~\bibnamefont
  {Chakraborty}}\ and\ \bibinfo {author} {\bibfnamefont {D.}~\bibnamefont
  {Mazziotti}},\ }\href {\doibase 10.1103/PhysRevA.91.010101} {\bibfield
  {journal} {\bibinfo  {journal} {Phys. Rev. A}\ }\textbf {\bibinfo {volume}
  {91}},\ \bibinfo {pages} {010101} (\bibinfo {year}
  {2015}{\natexlab{b}})}\BibitemShut {NoStop}%
\bibitem [{\citenamefont {Schilling}(2015{\natexlab{b}})}]{CS2015Hubbard}%
  \BibitemOpen
  \bibfield  {author} {\bibinfo {author} {\bibfnamefont {C.}~\bibnamefont
  {Schilling}},\ }\href {\doibase 10.1103/PhysRevB.92.155149} {\bibfield
  {journal} {\bibinfo  {journal} {Phys. Rev. B}\ }\textbf {\bibinfo {volume}
  {92}},\ \bibinfo {pages} {155149} (\bibinfo {year}
  {2015}{\natexlab{b}})}\BibitemShut {NoStop}%
\bibitem [{\citenamefont {Schilling}(2014{\natexlab{b}})}]{CSthesis}%
  \BibitemOpen
  \bibfield  {author} {\bibinfo {author} {\bibfnamefont {C.}~\bibnamefont
  {Schilling}},\ }\emph {\bibinfo {title} {Quantum marginal problem and its
  physical relevance}},\ \href {\doibase 10.3929/ethz-a-010139282} {Ph.D.
  thesis},\ \bibinfo  {school} {ETH-Z\"urich} (\bibinfo {year}
  {2014}{\natexlab{b}})\BibitemShut {NoStop}%
\bibitem [{\citenamefont {Klyachko}(2009)}]{Kly1}%
  \BibitemOpen
  \bibfield  {author} {\bibinfo {author} {\bibfnamefont {A.}~\bibnamefont
  {Klyachko}},\ }\href {http://arxiv.org/abs/0904.2009} {\bibfield  {journal}
  {\bibinfo  {journal} {arXiv:0904.2009}\ } (\bibinfo {year}
  {2009})}\BibitemShut {NoStop}%
\bibitem [{\citenamefont {Schilling}, \citenamefont {Benavides-Riveros},\ and\
  \citenamefont {Vrana}(2017)}]{SBV}%
  \BibitemOpen
  \bibfield  {author} {\bibinfo {author} {\bibfnamefont {C.}~\bibnamefont
  {Schilling}}, \bibinfo {author} {\bibfnamefont {C.~L.}\ \bibnamefont
  {Benavides-Riveros}}, \ and\ \bibinfo {author} {\bibfnamefont
  {P.}~\bibnamefont {Vrana}},\ }\href {https://arxiv.org/abs/1703.01612}
  {\bibfield  {journal} {\bibinfo  {journal} {arXiv:1703.01612}\ } (\bibinfo
  {year} {2017})}\BibitemShut {NoStop}%
\bibitem [{\citenamefont {Benavides-Riveros}(2015)}]{Benavidestesis}%
  \BibitemOpen
  \bibfield  {author} {\bibinfo {author} {\bibfnamefont {C.~L.}\ \bibnamefont
  {Benavides-Riveros}},\ }\emph {\bibinfo {title} {Disentangling the marginal
  problem in quantum chemistry}},\ \href
  {https://zaguan.unizar.es/record/47393?ln=es} {Ph.D. thesis},\ \bibinfo
  {school} {Universidad de Zaragoza} (\bibinfo {year} {2015})\BibitemShut
  {NoStop}%
\bibitem [{\citenamefont {Cleland}, \citenamefont {Booth},\ and\ \citenamefont
  {Alavi}(2010)}]{Cleland}%
  \BibitemOpen
  \bibfield  {author} {\bibinfo {author} {\bibfnamefont {D.}~\bibnamefont
  {Cleland}}, \bibinfo {author} {\bibfnamefont {G.~H.}\ \bibnamefont {Booth}},
  \ and\ \bibinfo {author} {\bibfnamefont {A.}~\bibnamefont {Alavi}},\ }\href
  {\doibase 10.1063/1.3302277} {\bibfield  {journal} {\bibinfo  {journal} {J.
  Chem. Phys.}\ }\textbf {\bibinfo {volume} {132}},\ \bibinfo {pages} {041103}
  (\bibinfo {year} {2010})}\BibitemShut {NoStop}%
\bibitem [{\citenamefont {Giner}, \citenamefont {Scemama},\ and\ \citenamefont
  {Caffarel}(2013)}]{Caffarel}%
  \BibitemOpen
  \bibfield  {author} {\bibinfo {author} {\bibfnamefont {E.}~\bibnamefont
  {Giner}}, \bibinfo {author} {\bibfnamefont {A.}~\bibnamefont {Scemama}}, \
  and\ \bibinfo {author} {\bibfnamefont {M.}~\bibnamefont {Caffarel}},\ }\href
  {\doibase 10.1139/cjc-2013-0017} {\bibfield  {journal} {\bibinfo  {journal}
  {Can. J. Chem.}\ }\textbf {\bibinfo {volume} {91}},\ \bibinfo {pages} {879}
  (\bibinfo {year} {2013})}\BibitemShut {NoStop}%
\bibitem [{\citenamefont {Caffarel}\ \emph {et~al.}(2016)\citenamefont
  {Caffarel}, \citenamefont {Applencourt}, \citenamefont {Giner},\ and\
  \citenamefont {Scemama}}]{CaffarelI}%
  \BibitemOpen
  \bibfield  {author} {\bibinfo {author} {\bibfnamefont {M.}~\bibnamefont
  {Caffarel}}, \bibinfo {author} {\bibfnamefont {T.}~\bibnamefont
  {Applencourt}}, \bibinfo {author} {\bibfnamefont {E.}~\bibnamefont {Giner}},
  \ and\ \bibinfo {author} {\bibfnamefont {A.}~\bibnamefont {Scemama}},\
  }\enquote {\bibinfo {title} {Recent progress in quantum monte carlo},}\ in\
  \href@noop {} {\emph {\bibinfo {booktitle} {Using CIPSI Nodes in Diffusion
  Monte Carlo}}}\ (\bibinfo  {publisher} {ACS},\ \bibinfo {year} {2016})\
  Chap.~\bibinfo {chapter} {2}, pp.\ \bibinfo {pages} {15--46}\BibitemShut
  {NoStop}%
\bibitem [{\citenamefont {Borland}\ and\ \citenamefont
  {Dennis}(1972)}]{Borl1972}%
  \BibitemOpen
  \bibfield  {author} {\bibinfo {author} {\bibfnamefont {R.}~\bibnamefont
  {Borland}}\ and\ \bibinfo {author} {\bibfnamefont {K.}~\bibnamefont
  {Dennis}},\ }\href {http://stacks.iop.org/0022-3700/5/i=1/a=009} {\bibfield
  {journal} {\bibinfo  {journal} {J. Phys. B}\ }\textbf {\bibinfo {volume}
  {5}},\ \bibinfo {pages} {7} (\bibinfo {year} {1972})}\BibitemShut {NoStop}%
\bibitem [{\citenamefont {L\"owdin}\ and\ \citenamefont {Shull}(1956)}]{LS}%
  \BibitemOpen
  \bibfield  {author} {\bibinfo {author} {\bibfnamefont {P.-O.}\ \bibnamefont
  {L\"owdin}}\ and\ \bibinfo {author} {\bibfnamefont {H.}~\bibnamefont
  {Shull}},\ }\href {\doibase 10.1103/PhysRev.101.1730} {\bibfield  {journal}
  {\bibinfo  {journal} {Phys. Rev.}\ }\textbf {\bibinfo {volume} {101}},\
  \bibinfo {pages} {1730} (\bibinfo {year} {1956})}\BibitemShut {NoStop}%
\bibitem [{\citenamefont {Kutzelnigg}\ and\ \citenamefont
  {Smith}(1968)}]{KS68}%
  \BibitemOpen
  \bibfield  {author} {\bibinfo {author} {\bibfnamefont {W.}~\bibnamefont
  {Kutzelnigg}}\ and\ \bibinfo {author} {\bibfnamefont {V.~H.}\ \bibnamefont
  {Smith}},\ }\href {\doibase 10.1002/qua.560020410} {\bibfield  {journal}
  {\bibinfo  {journal} {Int. J. Quant. Chem.}\ }\textbf {\bibinfo {volume}
  {2}},\ \bibinfo {pages} {531} (\bibinfo {year} {1968})}\BibitemShut {NoStop}%
\bibitem [{\citenamefont {Bytautas}\ \emph {et~al.}(2011)\citenamefont
  {Bytautas}, \citenamefont {Henderson}, \citenamefont {Jim\'enez-Hoyos},
  \citenamefont {Ellis},\ and\ \citenamefont {Scuseria}}]{Seniority}%
  \BibitemOpen
  \bibfield  {author} {\bibinfo {author} {\bibfnamefont {L.}~\bibnamefont
  {Bytautas}}, \bibinfo {author} {\bibfnamefont {T.~M.}\ \bibnamefont
  {Henderson}}, \bibinfo {author} {\bibfnamefont {C.~A.}\ \bibnamefont
  {Jim\'enez-Hoyos}}, \bibinfo {author} {\bibfnamefont {J.~K.}\ \bibnamefont
  {Ellis}}, \ and\ \bibinfo {author} {\bibfnamefont {G.~E.}\ \bibnamefont
  {Scuseria}},\ }\href {\doibase 10.1063/1.3613706} {\bibfield  {journal}
  {\bibinfo  {journal} {J. Chem. Phys.}\ }\textbf {\bibinfo {volume} {135}},\
  \bibinfo {pages} {044119} (\bibinfo {year} {2011})}\BibitemShut {NoStop}%
\bibitem [{\citenamefont {Bytautas}, \citenamefont {Scuseria},\ and\
  \citenamefont {Ruedenberg}(2015)}]{SeniorityI}%
  \BibitemOpen
  \bibfield  {author} {\bibinfo {author} {\bibfnamefont {L.}~\bibnamefont
  {Bytautas}}, \bibinfo {author} {\bibfnamefont {G.~E.}\ \bibnamefont
  {Scuseria}}, \ and\ \bibinfo {author} {\bibfnamefont {K.}~\bibnamefont
  {Ruedenberg}},\ }\href {\doibase 10.1063/1.4929904} {\bibfield  {journal}
  {\bibinfo  {journal} {J. Chem. Phys.}\ }\textbf {\bibinfo {volume} {143}},\
  \bibinfo {pages} {094105} (\bibinfo {year} {2015})}\BibitemShut {NoStop}%
\bibitem [{\citenamefont {L\'evay}\ and\ \citenamefont
  {Vrana}(2008)}]{Vrana2008}%
  \BibitemOpen
  \bibfield  {author} {\bibinfo {author} {\bibfnamefont {P.}~\bibnamefont
  {L\'evay}}\ and\ \bibinfo {author} {\bibfnamefont {P.}~\bibnamefont
  {Vrana}},\ }\href {\doibase 10.1103/PhysRevA.78.022329} {\bibfield  {journal}
  {\bibinfo  {journal} {Phys. Rev. A}\ }\textbf {\bibinfo {volume} {78}},\
  \bibinfo {pages} {022329} (\bibinfo {year} {2008})}\BibitemShut {NoStop}%
\bibitem [{\citenamefont {Mentel}\ \emph {et~al.}(2014)\citenamefont {Mentel},
  \citenamefont {van Meer}, \citenamefont {Gritsenko},\ and\ \citenamefont
  {Baerends}}]{Mentel2014}%
  \BibitemOpen
  \bibfield  {author} {\bibinfo {author} {\bibfnamefont {L.~M.}\ \bibnamefont
  {Mentel}}, \bibinfo {author} {\bibfnamefont {R.}~\bibnamefont {van Meer}},
  \bibinfo {author} {\bibfnamefont {O.~V.}\ \bibnamefont {Gritsenko}}, \ and\
  \bibinfo {author} {\bibfnamefont {E.~J.}\ \bibnamefont {Baerends}},\ }\href
  {http://scitation.aip.org/content/aip/journal/jcp/140/21/10.1063/1.4879776}
  {\bibfield  {journal} {\bibinfo  {journal} {J. Chem. Phys.}\ }\textbf
  {\bibinfo {volume} {140}},\ \bibinfo {pages} {214105} (\bibinfo {year}
  {2014})}\BibitemShut {NoStop}%
\bibitem [{\citenamefont {Shenvi}\ and\ \citenamefont
  {Izmaylov}(2010)}]{acrdm}%
  \BibitemOpen
  \bibfield  {author} {\bibinfo {author} {\bibfnamefont {N.}~\bibnamefont
  {Shenvi}}\ and\ \bibinfo {author} {\bibfnamefont {A.~F.}\ \bibnamefont
  {Izmaylov}},\ }\href {\doibase 10.1103/PhysRevLett.105.213003} {\bibfield
  {journal} {\bibinfo  {journal} {Phys. Rev. Lett.}\ }\textbf {\bibinfo
  {volume} {105}},\ \bibinfo {pages} {213003} (\bibinfo {year}
  {2010})}\BibitemShut {NoStop}%
\bibitem [{\citenamefont {Horodecki}\ \emph {et~al.}(2009)\citenamefont
  {Horodecki}, \citenamefont {Horodecki}, \citenamefont {Horodecki},\ and\
  \citenamefont {Horodecki}}]{Horo}%
  \BibitemOpen
  \bibfield  {author} {\bibinfo {author} {\bibfnamefont {R.}~\bibnamefont
  {Horodecki}}, \bibinfo {author} {\bibfnamefont {P.}~\bibnamefont
  {Horodecki}}, \bibinfo {author} {\bibfnamefont {M.}~\bibnamefont
  {Horodecki}}, \ and\ \bibinfo {author} {\bibfnamefont {K.}~\bibnamefont
  {Horodecki}},\ }\href {\doibase 10.1103/RevModPhys.81.865} {\bibfield
  {journal} {\bibinfo  {journal} {Rev. Mod. Phys.}\ }\textbf {\bibinfo {volume}
  {81}},\ \bibinfo {pages} {865} (\bibinfo {year} {2009})}\BibitemShut
  {NoStop}%
\bibitem [{\citenamefont {Szalay}(2015)}]{PhysRevA.92.042329}%
  \BibitemOpen
  \bibfield  {author} {\bibinfo {author} {\bibfnamefont {S.}~\bibnamefont
  {Szalay}},\ }\href {\doibase 10.1103/PhysRevA.92.042329} {\bibfield
  {journal} {\bibinfo  {journal} {Phys. Rev. A}\ }\textbf {\bibinfo {volume}
  {92}},\ \bibinfo {pages} {042329} (\bibinfo {year} {2015})}\BibitemShut
  {NoStop}%
\bibitem [{\citenamefont {Shimony}(1995)}]{Shimony}%
  \BibitemOpen
  \bibfield  {author} {\bibinfo {author} {\bibfnamefont {A.}~\bibnamefont
  {Shimony}},\ }\href {\doibase 10.1111/j.1749-6632.1995.tb39008.x} {\bibfield
  {journal} {\bibinfo  {journal} {Ann. N. Y. Acad. Sci.}\ }\textbf {\bibinfo
  {volume} {755}},\ \bibinfo {pages} {675} (\bibinfo {year}
  {1995})}\BibitemShut {NoStop}%
\bibitem [{\citenamefont {Myers}\ and\ \citenamefont {Wu}(2010)}]{Myers2010}%
  \BibitemOpen
  \bibfield  {author} {\bibinfo {author} {\bibfnamefont {J.~M.}\ \bibnamefont
  {Myers}}\ and\ \bibinfo {author} {\bibfnamefont {T.~T.}\ \bibnamefont {Wu}},\
  }\href {\doibase 10.1007/s11128-009-0146-5} {\bibfield  {journal} {\bibinfo
  {journal} {Quantum Inf. Process.}\ }\textbf {\bibinfo {volume} {9}},\
  \bibinfo {pages} {239} (\bibinfo {year} {2010})}\BibitemShut {NoStop}%
\bibitem [{\citenamefont {Zhang}\ and\ \citenamefont {Kollar}(2014)}]{Zhan}%
  \BibitemOpen
  \bibfield  {author} {\bibinfo {author} {\bibfnamefont {J.~M.}\ \bibnamefont
  {Zhang}}\ and\ \bibinfo {author} {\bibfnamefont {M.}~\bibnamefont {Kollar}},\
  }\href {\doibase 10.1103/PhysRevA.89.012504} {\bibfield  {journal} {\bibinfo
  {journal} {Phys. Rev. A}\ }\textbf {\bibinfo {volume} {89}},\ \bibinfo
  {pages} {012504} (\bibinfo {year} {2014})}\BibitemShut {NoStop}%
\bibitem [{\citenamefont {Smith}(1966)}]{SmithDarwin}%
  \BibitemOpen
  \bibfield  {author} {\bibinfo {author} {\bibfnamefont {D.~W.}\ \bibnamefont
  {Smith}},\ }\href {\doibase 10.1103/PhysRev.147.896} {\bibfield  {journal}
  {\bibinfo  {journal} {Phys. Rev.}\ }\textbf {\bibinfo {volume} {147}},\
  \bibinfo {pages} {896} (\bibinfo {year} {1966})}\BibitemShut {NoStop}%
\bibitem [{\citenamefont {Schmidt}\ \emph {et~al.}(1993)\citenamefont
  {Schmidt}, \citenamefont {Baldridge}, \citenamefont {Boatz}, \citenamefont
  {Elbert}, \citenamefont {Gordon}, \citenamefont {Jensen}, \citenamefont
  {Koseki}, \citenamefont {Matsunaga}, \citenamefont {Nguyen}, \citenamefont
  {Su}, \citenamefont {Windus}, \citenamefont {Dupuis},\ and\ \citenamefont
  {Montgomery}}]{Gamess}%
  \BibitemOpen
  \bibfield  {author} {\bibinfo {author} {\bibfnamefont {M.~W.}\ \bibnamefont
  {Schmidt}}, \bibinfo {author} {\bibfnamefont {K.~K.}\ \bibnamefont
  {Baldridge}}, \bibinfo {author} {\bibfnamefont {J.~A.}\ \bibnamefont
  {Boatz}}, \bibinfo {author} {\bibfnamefont {S.~T.}\ \bibnamefont {Elbert}},
  \bibinfo {author} {\bibfnamefont {M.~S.}\ \bibnamefont {Gordon}}, \bibinfo
  {author} {\bibfnamefont {J.~H.}\ \bibnamefont {Jensen}}, \bibinfo {author}
  {\bibfnamefont {S.}~\bibnamefont {Koseki}}, \bibinfo {author} {\bibfnamefont
  {N.}~\bibnamefont {Matsunaga}}, \bibinfo {author} {\bibfnamefont {K.~A.}\
  \bibnamefont {Nguyen}}, \bibinfo {author} {\bibfnamefont {S.}~\bibnamefont
  {Su}}, \bibinfo {author} {\bibfnamefont {T.~L.}\ \bibnamefont {Windus}},
  \bibinfo {author} {\bibfnamefont {M.}~\bibnamefont {Dupuis}}, \ and\ \bibinfo
  {author} {\bibfnamefont {J.~A.}\ \bibnamefont {Montgomery}},\ }\href
  {\doibase 10.1002/jcc.540141112} {\bibfield  {journal} {\bibinfo  {journal}
  {J. Comput. Chem.}\ }\textbf {\bibinfo {volume} {14}},\ \bibinfo {pages}
  {1347} (\bibinfo {year} {1993})}\BibitemShut {NoStop}%
\bibitem [{\citenamefont {Coulson}\ and\ \citenamefont
  {Fischer}(1949)}]{CouFisch}%
  \BibitemOpen
  \bibfield  {author} {\bibinfo {author} {\bibfnamefont {C.~A.}\ \bibnamefont
  {Coulson}}\ and\ \bibinfo {author} {\bibfnamefont {I.}~\bibnamefont
  {Fischer}},\ }\href {\doibase 10.1080/14786444908521726} {\bibfield
  {journal} {\bibinfo  {journal} {Philos. Mag.}\ }\textbf {\bibinfo {volume}
  {40}},\ \bibinfo {pages} {386} (\bibinfo {year} {1949})}\BibitemShut
  {NoStop}%
\bibitem [{\citenamefont {Murmann}\ \emph {et~al.}(2015)\citenamefont
  {Murmann}, \citenamefont {Bergschneider}, \citenamefont {Klinkhamer},
  \citenamefont {Z\"urn}, \citenamefont {Lompe},\ and\ \citenamefont
  {Jochim}}]{PhysRevLett114}%
  \BibitemOpen
  \bibfield  {author} {\bibinfo {author} {\bibfnamefont {S.}~\bibnamefont
  {Murmann}}, \bibinfo {author} {\bibfnamefont {A.}~\bibnamefont
  {Bergschneider}}, \bibinfo {author} {\bibfnamefont {V.~M.}\ \bibnamefont
  {Klinkhamer}}, \bibinfo {author} {\bibfnamefont {G.}~\bibnamefont {Z\"urn}},
  \bibinfo {author} {\bibfnamefont {T.}~\bibnamefont {Lompe}}, \ and\ \bibinfo
  {author} {\bibfnamefont {S.}~\bibnamefont {Jochim}},\ }\href {\doibase
  10.1103/PhysRevLett.114.080402} {\bibfield  {journal} {\bibinfo  {journal}
  {Phys. Rev. Lett.}\ }\textbf {\bibinfo {volume} {114}},\ \bibinfo {pages}
  {080402} (\bibinfo {year} {2015})}\BibitemShut {NoStop}%
\bibitem [{\citenamefont {Chai}(2012)}]{Chai}%
  \BibitemOpen
  \bibfield  {author} {\bibinfo {author} {\bibfnamefont {J.-D.}\ \bibnamefont
  {Chai}},\ }\href
  {http://scitation.aip.org/content/aip/journal/jcp/136/15/10.1063/1.3703894}
  {\bibfield  {journal} {\bibinfo  {journal} {J. Chem. Phys.}\ }\textbf
  {\bibinfo {volume} {136}},\ \bibinfo {pages} {154104} (\bibinfo {year}
  {2012})}\BibitemShut {NoStop}%
\bibitem [{\citenamefont {Boguslawski}\ \emph {et~al.}(2013)\citenamefont
  {Boguslawski}, \citenamefont {Tecmer}, \citenamefont {Barcza}, \citenamefont
  {Legeza},\ and\ \citenamefont {Reiher}}]{doi:10.1021/ct400247p}%
  \BibitemOpen
  \bibfield  {author} {\bibinfo {author} {\bibfnamefont {K.}~\bibnamefont
  {Boguslawski}}, \bibinfo {author} {\bibfnamefont {P.}~\bibnamefont {Tecmer}},
  \bibinfo {author} {\bibfnamefont {G.}~\bibnamefont {Barcza}}, \bibinfo
  {author} {\bibfnamefont {{\"{O}}.}~\bibnamefont {Legeza}}, \ and\ \bibinfo
  {author} {\bibfnamefont {M.}~\bibnamefont {Reiher}},\ }\href {\doibase
  10.1021/ct400247p} {\bibfield  {journal} {\bibinfo  {journal} {J. Chem.
  Theory Comput.}\ }\textbf {\bibinfo {volume} {9}},\ \bibinfo {pages} {2959}
  (\bibinfo {year} {2013})}\BibitemShut {NoStop}%
\bibitem [{\citenamefont {Friis}(2016)}]{Friis}%
  \BibitemOpen
  \bibfield  {author} {\bibinfo {author} {\bibfnamefont {N.}~\bibnamefont
  {Friis}},\ }\href {http://stacks.iop.org/1367-2630/18/i=3/a=033014}
  {\bibfield  {journal} {\bibinfo  {journal} {New J. Phys.}\ }\textbf {\bibinfo
  {volume} {18}},\ \bibinfo {pages} {033014} (\bibinfo {year}
  {2016})}\BibitemShut {NoStop}%
\end{thebibliography}%

\end{document}